\newif\ifdraft  \draftfalse   
\newif\ifpdflatex  \pdflatexfalse   
\ifpdflatex
  \documentclass[a4paper,UKenglish]{lipics}
\else
  \documentclass[a4paper,12pt]{article}
  \usepackage[a4paper,margin=1.5in,bottom=1.25in,top=1.25in,marginpar=1in]{geometry}
\fi
\usepackage{datetime}
\longdate

\usepackage[T1]{fontenc}
\title{Ultimate periodicity of b-recognisable sets : \\ a quasilinear procedure}
\author
{
        Victor Marsault\thanks{Corresponding author, LTCI, Telecom ParisTech}~
        and 
        Jacques Sakarovitch\thanks{LTCI, CNRS / Telecom ParisTech}
}
\longdate
\date{\today}

\usepackage{xspace}
\usepackage{vaucanson-g-O}

\usepackage{enumitem}
\usepackage{amssymb,amsmath}
\usepackage{amsthm}
\usepackage{amsfonts}
\usepackage[british,UKenglish,USenglish,english,american]{babel}
\ifdraft%
  \usepackage[notref,notcite]{showkeys}
  
\fi

\newcommand{\BIBINPUTSDIR}{}
\usepackage[mathlines,displaymath]{lineno}
\nolinenumbers
\ifdraft\linenumbers\fi
\newcommand{\vmfigure}[2]{
  \begin{figure}[ht!]
    \centering
    \input{#1/#2.label}
    \noindent\makebox[\textwidth]{\includegraphics{#1/#2}}
    \ifthenelse{\equal{\curlang}{fr}}
    {\caption{\frcurrentcaption}}{}
    \ifthenelse{\equal{\curlang}{en}}
    {\caption{\encurrentcaption}}{}%
    \def\vmlastfigure{#2}
    \lfigure{#2}
  \end{figure}
}

\usepackage{ifthen}
\usepackage{etoolbox}
\usepackage{xargs}
\usepackage{slantsc}

\newcommandx{\newtheoremy}[3][2={}]{
  \ifthenelse{\equal{#2}{}}{
    \ifcsmacro{#1}{}{\newtheorem{#1}{#3}}
  }{
    \ifcsmacro{#1}{}{\newtheorem{#1}[#2]{#3}}
  }
}

\newtheoremy{thm}{thm}
\newcommand{\thmBlockFont}[1]{\sc{#1}}
\newcommand{\RefSeparator}{.}
\newcommand{\generalref}[2]{\ref{#1\RefSeparator#2}}
\newcommand{\generalpageref}[2]{\pageref{#1\RefSeparator#2}}
\newcommand{\generallabel}[2]{\label{#1\RefSeparator#2}}
\newcommand{\PageName}{page}

\newcommand{\AlgorithmName}{Algorithm}
\newcommand{\AlgorithmRefName}{Algorithm}
\newcommand{\AlgorithmRefPrefix}{a}
\newtheoremy{algorithm}[thm]{\thmBlockFont{\AlgorithmName}}

\makeatletter
\newcommand*{\ralgorithm}{\@ifstar{\generalref{\AlgorithmRefPrefix}}{\AlgorithmRefName~\ralgorithm*}}
\newcommand*{\palgorithm}{\@ifstar{\generalpageref{\AlgorithmRefPrefix}}{\PageName~\palgorithm*}}
\makeatother

\newcommand{\CorollaryName}{Corollary}
\newcommand{\CorollaryRefName}{Corollary}
\newcommand{\CorollaryRefPrefix}{c}
\newtheoremy{corollary}[thm]{\thmBlockFont{\CorollaryName}}

\makeatletter
\newcommand*{\rcorollary}{\@ifstar{\generalref{\CorollaryRefPrefix}}{\CorollaryRefName~\rcorollary*}}
\newcommand*{\pcorollary}{\@ifstar{\generalpageref{\CorollaryRefPrefix}}{\PageName~\pcorollary*}}
\makeatother

\newcommand{\DefinitionName}{Definition}
\newcommand{\DefinitionRefName}{Definition}
\newcommand{\DefinitionRefPrefix}{d}
\newtheoremy{definition}[thm]{\thmBlockFont{\DefinitionName}}

\makeatletter
\newcommand*{\rdefinition}{\@ifstar{\generalref{\DefinitionRefPrefix}}{\DefinitionRefName~\rdefinition*}}
\newcommand*{\pdefinition}{\@ifstar{\generalpageref{\DefinitionRefPrefix}}{\PageName~\pdefinition*}}
\makeatother

\newcommand{\ExampleName}{Example}
\newcommand{\ExampleRefName}{Example}
\newcommand{\ExampleRefPrefix}{e}
\newtheoremy{example}[thm]{\thmBlockFont{\ExampleName}}

\makeatletter
\newcommand*{\rexample}{\@ifstar{\generalref{\ExampleRefPrefix}}{\ExampleRefName~\rexample*}}
\newcommand*{\pexample}{\@ifstar{\generalpageref{\ExampleRefPrefix}}{\PageName~\pexample*}}
\makeatother

\newcommand{\LemmaName}{Lemma}
\newcommand{\LemmaRefName}{Lemma}
\newcommand{\LemmaRefPrefix}{l}
\newtheoremy{lemma}[thm]{\thmBlockFont{\LemmaName}}
\newcommand{\llemma}[1]{\generallabel{\LemmaRefPrefix}{#1}}
\makeatletter
\newcommand*{\rlemma}{\@ifstar{\generalref{\LemmaRefPrefix}}{\LemmaRefName~\rlemma*}}
\newcommand*{\plemmam}{\@ifstar{\generalpageref{\LemmaRefPrefix}}{\PageName~\plemma*}}
\makeatother

\newcommand{\PropositionName}{Proposition}
\newcommand{\PropositionRefName}{Proposition}
\newcommand{\PropositionRefPrefix}{p}
\newtheoremy{proposition}[thm]{\thmBlockFont{\PropositionName}}
\newcommand{\lproposition}[1]{\generallabel{\PropositionRefPrefix}{#1}}
\makeatletter
\newcommand*{\rproposition}{\@ifstar{\generalref{\PropositionRefPrefix}}{\PropositionRefName~\rproposition*}}
\newcommand*{\pproposition}{\@ifstar{\generalpageref{\PropositionRefPrefix}}{\PageName~\pproposition*}}
\makeatother

\newcommand{\PropertyName}{Property}
\newcommand{\PropertyRefName}{Property}
\newcommand{\PropertyRefPrefix}{pp}
\newtheoremy{property}[thm]{\thmBlockFont{\PropertyName}}

\makeatletter
\newcommand*{\rproperty}{\@ifstar{\generalref{\PropertyRefPrefix}}{\PropertyRefName~\rproperty*}}
\newcommand*{\pproperty}{\@ifstar{\generalpageref{\PropertyRefPrefix}}{\PageName~\pproperty*}}
\makeatother

\newcommand{\QuestionName}{Question}
\newcommand{\QuestionRefName}{Question}
\newcommand{\QuestionRefPrefix}{q}
\newtheoremy{question}[thm]{\thmBlockFont{\QuestionName}}

\makeatletter
\newcommand*{\rquestion}{\@ifstar{\generalref{\QuestionRefPrefix}}{\QuestionRefName~\rquestion*}}
\newcommand*{\pquestion}{\@ifstar{\generalpageref{\QuestionRefPrefix}}{\PageName~\pquestion*}}
\makeatother

\newcommand{\RemarkName}{Remark}
\newcommand{\RemarkRefName}{Remark}
\newcommand{\RemarkRefPrefix}{r}
\newtheoremy{remark}[thm]{\thmBlockFont{\RemarkName}}
\newcommand{\lremark}[1]{\generallabel{\RemarkRefPrefix}{#1}}
\makeatletter
\newcommand*{\rremark}{\@ifstar{\generalref{\RemarkRefPrefix}}{\RemarkRefName~\rremark*}}
\newcommand*{\premark}{\@ifstar{\generalpageref{\RemarkRefPrefix}}{\PageName~\premark*}}
\makeatother

\newcommand{\NotationName}{Notation}
\newcommand{\NotationRefName}{Notation}
\newcommand{\NotationRefPrefix}{n}
\newtheoremy{notation}[thm]{\thmBlockFont{\NotationName}}

\makeatletter
\newcommand*{\rnotation}{\@ifstar{\generalref{\NotationRefPrefix}}{\NotationRefName~\rnotation*}}
\newcommand*{\pnotation}{\@ifstar{\generalpageref{\NotationRefPrefix}}{\PageName~\pnotation*}}
\makeatother

\newcommand{\TheoremName}{Theorem}
\newcommand{\TheoremRefName}{Theorem}
\newcommand{\TheoremRefPrefix}{t}
\newtheoremy{theorem}[thm]{\thmBlockFont{\TheoremName}}
\newtheoremy{Theorem}{\thmBlockFont{\TheoremName}}

\newcommand{\ltheorem}[1]{\generallabel{\TheoremRefPrefix}{#1}}
\makeatletter
\newcommand*{\rtheorem}{\@ifstar{\generalref{\TheoremRefPrefix}}{\TheoremRefName~\rtheorem*}}
\newcommand*{\ptheorem}{\@ifstar{\generalpageref{\TheoremRefPrefix}}{\PageName~\ptheorem*}}
\makeatother

\newcommand{\FigureRefName}{Figure}
\newcommand{\FigureRefPrefix}{f}
\newcommand{\lfigure}[1]{\generallabel{\FigureRefPrefix}{#1}}
\makeatletter
\newcommand*{\rfigure}{\@ifstar{\generalref{\FigureRefPrefix}}{\FigureRefName~\rfigure*}}
\newcommand*{\pfigure}{\@ifstar{\generalpageref{\FigureRefPrefix}}{\PageName~\pfigure*}}
\makeatother

\newcommand{\EquationRefName}{Equation}
\newcommand{\EquationRefPrefix}{eq}

\makeatletter
\newcommand*{\requation}{\@ifstar{\generalref{\EquationRefPrefix}}{\EquationRefName~\requation*}}
\newcommand*{\pequation}{\@ifstar{\generalpageref{\EquationRefPrefix}}{\PageName~\pequation*}}
\makeatother

\newcommand{\SectionRefName}{Section}
\newcommand{\SectionRefPrefix}{s}
\newcommand{\lsection}[1]{\generallabel{\SectionRefPrefix}{#1}}
\makeatletter
\newcommand*{\rsection}{\@ifstar{\generalref{\SectionRefPrefix}}{\SectionRefName~\rsection*}}
\newcommand*{\psection}{\@ifstar{\generalpageref{\SectionRefPrefix}}{\PageName~\psection*}}
\makeatother

\newcommand{\ProblemName}{Problem}
\newcommand{\ProblemRefName}{Problem}
\newcommand{\ProblemRefPrefix}{pb}
\newtheoremy{problem}[thm]{\thmBlockFont{\ProblemName}}

\makeatletter
\newcommand*{\rproblem}{\@ifstar{\generalref{\ProblemRefPrefix}}{\ProblemRefName~\rproblem*}}
\newcommand*{\pproblem}{\@ifstar{\generalpageref{\ProblemRefPrefix}}{\PageName~\pproblem*}}
\makeatother

\newcommand{\vmrefprefix}[1]{%
  \ifthenelse{\equal{#1}{corollary}}{\CorollaryRefPrefix}{}%
  \ifthenelse{\equal{#1}{definition}}{\DefinitionRefPrefix}{}%
  \ifthenelse{\equal{#1}{example}}{\ExampleRefPrefix}{}%
  \ifthenelse{\equal{#1}{lemma}}{\LemmaRefPrefix}{}%
  \ifthenelse{\equal{#1}{proposition}}{\PropositionRefPrefix}{}%
  \ifthenelse{\equal{#1}{property}}{\PropertyRefPrefix}{}%
  \ifthenelse{\equal{#1}{question}}{\QuestionRefPrefix}{}%
  \ifthenelse{\equal{#1}{remark}}{\RemarkRefPrefix}{}%
  \ifthenelse{\equal{#1}{notation}}{\NotationRefPrefix}{}%
  \ifthenelse{\equal{#1}{theorem}}{\TheoremRefPrefix}{}%
  \ifthenelse{\equal{#1}{figure}}{\FigureRefPrefix}{}%
  \ifthenelse{\equal{#1}{equation}}{\EquationRefPrefix}{}%
  \ifthenelse{\equal{#1}{section}}{\SectionRefPrefix}{}%
}

\newcommand{\vmrefname}[1]{
  \ifthenelse{\equal{#1}{corollary}}{\CorollaryRefName}{}%
  \ifthenelse{\equal{#1}{definition}}{\DefinitionRefName}{}%
  \ifthenelse{\equal{#1}{example}}{\ExampleRefName}{}%
  \ifthenelse{\equal{#1}{lemma}}{\LemmaRefName}{}%
  \ifthenelse{\equal{#1}{proposition}}{\PropositionRefName}{}%
  \ifthenelse{\equal{#1}{property}}{\PropertyRefName}{}%
  \ifthenelse{\equal{#1}{question}}{\QuestionRefName}{}%
  \ifthenelse{\equal{#1}{remark}}{\RemarkRefName}{}%
  \ifthenelse{\equal{#1}{notation}}{\NotationRefName}{}%
  \ifthenelse{\equal{#1}{theorem}}{\TheoremRefName}{}%
  \ifthenelse{\equal{#1}{figure}}{\FigureRefName}{}%
  \ifthenelse{\equal{#1}{equation}}{\EquationRefName}{}%
  \ifthenelse{\equal{#1}{section}}{\SectionRefName}{}%
}

\def\jscompatibility{0}
\def\curlang{en}

\newcommandx{\vmnewcommandx}[5][2=0,3={},5={},usedefault]{
      \ifthenelse{\equal{\jscompatibility}{0}}
      {\newcommandx{#1}[#2][#3]{#4}} 
      {\newcommandx{#1}[#2][#3]{#5}} 
}

\newcommand{\word}[1]{\textup{\textrm{"}} #1 \textup{\textrm{"}}}

\vmnewcommandx{\wlen}[1]{|#1|}
\vmnewcommandx{\pathx}[1]{\nlb\xrightarrow{\ #1 \ }\nlb}
\vmnewcommandx{\pathy}[1]{\xleftarrow{\ #1 \ }}
\vmnewcommandx{\cod}[1]{\langle #1 \rangle}
\vmnewcommandx{\floor}[1]{\lfloor #1 \rfloor}
\vmnewcommandx{\ceil}[1]{\lceil #1 \rceil}

\newcommandx{\newcommandy}[5][1=i,3=0,4={}]{%
  \ifthenelse{\isundefined{#2}}{\newcommandx{#2}[#3][#4]{#5}}{%
      \ifthenelse{\equal{#1}{i}}{}{}%
      \ifthenelse{\equal{#1}{o}}{\renewcommandx{#2}[#3][#4]{#5}}{}%
    }%
}

\newcommand{\val}[1]{\widebar{#1}}

\newcommand{\ssc}[1]{\textbf{\textsc{#1}}}

\newcommand{\set}[1]{\{#1\}}

\newcommand{\Z}{\mathbb{Z}}
\newcommand{\N}{\mathbb{N}}

\newcommand{\widebar}{\overline}

\newcommandy[o]{\Cup}{\bigcup}
\newcommand{\nlb}{\nolinebreak}

\renewcommand{\thmBlockFont}[1]{\ssc{#1}}

\newcommand{\frcurrentcaption}{}
\newcommand{\encurrentcaption}{}
\newcommandy{\input{/.label}}[2]{
  \begin{figure}[ht!]
    \centering
    \input{#1/#2}
    \ifthenelse{\equal{\curlang}{fr}}
    {\caption{\frcurrentcaption}}{}
    \ifthenelse{\equal{\curlang}{en}}
    {\caption{\encurrentcaption}}{}%
    \def\vmlastfigure{#2}
    \lfigure{#2}
  \end{figure}
}
\renewcommand{\leq}{\leqslant}
\renewcommand{\geq}{\geqslant}
\renewcommand{\phi}{\varphi}
\renewcommand{\epsilon}{\varepsilon}
\renewcommand{\mod}{\text{~mod~}}

\newcommand{\fa}{\forall}
\newcommand{\ext}{\exists}

\newcommand{\e}{\text{\quad}}                 
\newcommand{\ee}{\text{\qquad}}               
\newcommand{\eee}{\text{\qquad \qquad}} 
\newsavebox{\InterSymbolSpace}
\savebox{\InterSymbolSpace}{\hspace{0.125em}}
\newsavebox{\SideFormulaSpace}
\savebox{\SideFormulaSpace}{\hspace{0.2em}}
\newcommand{\msp}{\usebox{\SideFormulaSpace}} 
\newcommand{\xmd}{\usebox{\InterSymbolSpace}} 
\newcommand{\eqpnt}{\makebox[0pt][l]{\: .}}

\newcommand{\quantvrg}{\, , \;}
\newcommand{\quantsp}{\ee }
\newcommand{\quantsmsp}{\e }
\newcommand{\LatinLocution}[1]{{\itshape #1}\xspace}

\newcommand{\cf}{\LatinLocution{cf.}}

\newcommand{\ie}{{that is, }}

\newcommand{\Nmbb}{\mathbb{N}}

\newcommand{\Zmbb}{\mathbb{Z}}
\newcommand{\UNmbb}{{\mathchoice
{\hbox{$\textstyle\rm 1\kern-0.2em I$}}%
{\hbox{$\textstyle\rm 1\kern-0.2em I$}}%
{\hbox{$\scriptstyle\rm 1\kern-0.15em I$}}%
{\hbox{$\scriptscriptstyle\rm 1\kern-0.1em I$}}%
}}
\newcommand{\Ac}{\mathcal{A}}
\newcommand{\Bc}{\mathcal{B}}
\newcommand{\Cc}{\mathcal{C}}
\newcommand{\Dc}{\mathcal{D}}

\newcommand{\Gc}{\mathcal{G}}

\newcommand{\Pc}{\mathcal{P}}

\newcommand{\Sc}{\mathcal{S}}

\newlength{\ArrowDiagSize}
\newlength{\ArrowDiagWidth}
\newpsstyle{SLDiagStyle}%
   {colsep=6ex,rowsep=5ex,nodesep=1ex,npos=.45,%
    arrows=->,linewidth=\ArrowDiagWidth,arrowsize=\ArrowDiagSize,%
        linestyle=solid,linecolor=\ArrowDiagColor}

\newenvironment{SLDiag}%
   {\psset{style=SLDiagStyle}\begin{psmatrix}}%
   {\end{psmatrix}}%
\newcommand{\CDSL}{\begin{SLDiag}}
\newcommand{\CDSLF}{\end{SLDiag}}
\newenvironment{DiagraBig}%
{\psmatrix[colsep=7ex,rowsep=6ex,arrows=->,nodesep=1ex,npos=.45]}%
{\endpsmatrix}
\newcommand{\CDB}{\begin{DiagraBig}}
\newcommand{\CDBF}{\end{DiagraBig}}
\newenvironment{DiagraSmall}%
{\psmatrix[colsep=3ex,rowsep=3ex,arrows=->,nodesep=1ex,npos=.45]}%
{\endpsmatrix}
\newcommand{\CDS}{\begin{DiagraSmall}}
\newcommand{\CDSF}{\end{DiagraSmall}}

\newcommand{\matriceuu}[1]%
    {\begin{pmatrix} #1 \end{pmatrix}}
\newcommand{\matricedd}[4]%
    {\begin{pmatrix} #1 & #2 \\ #3 & #4 \end{pmatrix}}
\newcommand{\vecteurd}[2]%
    {\begin{pmatrix} #1 \\ #2 \end{pmatrix}}
\newcommand{\ligned}[2]%
    {\begin{pmatrix} #1 & #2 \end{pmatrix}}
\newcommand{\matricett}[9]%
    {\begin{pmatrix}  #1 & #2 & #3 \\
                      #4 & #5 & #6 \\
                      #7 & #8 & #9 \end{pmatrix}}
\newcommand{\vecteurt}[3]%
    {\begin{pmatrix} #1 \\ #2 \\ #3 \end{pmatrix}}
\newcommand{\lignet}[3]%
    {\begin{pmatrix} #1 & #2 & #3 \end{pmatrix}}

\newlength{\jsWidthCol}
\newlength{\blocinterligne}
\newlength{\blocinterligned}
\newlength{\temparraycolsep}
\newlength{\longueurbloc}
\newlength{\hauteurbloc}
\newlength{\centragebloc}
\newlength{\longueurblc}
\newlength{\hauteurblc}
\newlength{\centrageblc}
\newcommand{\blocligne}[1]%
    {\framebox[\longueurbloc]{$#1$}}
\newcommand{\blocmatrice}[1]%
    {\framebox[\longueurbloc]{\rule[\centragebloc]{0mm}{\hauteurbloc}$#1$}}
\newcommand{\blocvecteur}[1]%
    {\framebox{\rule[\centragebloc]{0mm}{\hauteurbloc}$#1$}}
\newcommand{\blcligne}[1]%
    {\framebox[\longueurblc]{$#1$}}
\newcommand{\blcmatrice}[1]%
    {\framebox[\longueurblc]{\rule[\centrageblc]{0mm}{\hauteurblc}$#1$}}
\newcommand{\blcvecteur}[1]%
    {\framebox{\rule[\centrageblc]{0mm}{\hauteurblc}$#1$}}
\newcommand{\matriceddblvs}[4]
   {\setlength{\temparraycolsep}{\arraycolsep}%
    \setlength{\arraycolsep}{1.3pt}%
        \left (%
    \begin{array}{cc}%
                #1  & \blcligne{#2} \\
            \blcvecteur{#3} & \blcmatrice{#4}
        \end{array}%
        \right )%
    \setlength{\arraycolsep}{\temparraycolsep}%
   }%
\newcommand{\vecteurdblvs}[2]%
   {\setlength{\temparraycolsep}{\arraycolsep}%
    \setlength{\arraycolsep}{1.5pt}%
        \left (%
    \begin{array}{c}%
                #1  \\
            \blcvecteur{#2}
        \end{array}%
        \right )%
    \setlength{\arraycolsep}{\temparraycolsep}%
   }%
\newcommand{\lignedblvs}[2]%
   {\setlength{\temparraycolsep}{\arraycolsep}%
    \setlength{\arraycolsep}{1.5pt}%
        \left (%
    \begin{array}{cc}%
                #1  & \blcligne{#2}
        \end{array}%
        \right )%
    \setlength{\arraycolsep}{\temparraycolsep}%
   }%
\newcommand{\matricettblvs}[9]
   {\setlength{\temparraycolsep}{\arraycolsep}%
    \setlength{\arraycolsep}{1.5pt}%
        \left (%
    \begin{array}{ccc}%
                #1  & \blcligne{#2} & #3\\
            \blcvecteur{#4} & \blcmatrice{#5} & \blcvecteur{#6}\\
                #7  & \blcligne{#8} & #9\\
        \end{array}%
        \right )%
    \setlength{\arraycolsep}{\temparraycolsep}%
   }%
\newcommand{\vecteurtblvs}[3]%
   {\setlength{\temparraycolsep}{\arraycolsep}%
    \setlength{\arraycolsep}{1.5pt}%
        \left (%
    \begin{array}{c}%
                #1  \\
            \blcvecteur{#2}\\
                #3
        \end{array}%
        \right )%
    \setlength{\arraycolsep}{\temparraycolsep}%
   }%
\newcommand{\lignetblvs}[3]%
   {\setlength{\temparraycolsep}{\arraycolsep}%
    \setlength{\arraycolsep}{1.5pt}%
        \left (%
    \begin{array}{ccc}%
                #1  & \blcligne{#2} & #3
        \end{array}%
        \right )%
    \setlength{\arraycolsep}{\temparraycolsep}%
   }%
\newcommand{\matricettblblvs}[9]
   {\setlength{\temparraycolsep}{\arraycolsep}%
    \setlength{\arraycolsep}{1.5pt}%
        \left (%
    \begin{array}{ccc}%
                #1  & \blcligne{#2} & \blcligne{#3}\\
            \blcvecteur{#4} & \blcmatrice{#5} & \blcmatrice{#6}\\
                \blcvecteur{#7}  & \blcmatrice{#8} & \blcmatrice{#9}\\
        \end{array}%
        \right )%
    \setlength{\arraycolsep}{\temparraycolsep}%
   }%
\newcommand{\vecteurtblblvs}[3]%
   {\setlength{\temparraycolsep}{\arraycolsep}%
    \setlength{\arraycolsep}{1.5pt}%
        \left (%
    \begin{array}{c}%
                #1  \\
            \blcvecteur{#2}\\
                \blcvecteur{#3}
        \end{array}%
        \right )%
    \setlength{\arraycolsep}{\temparraycolsep}%
   }%
\newcommand{\lignetblblvs}[3]%
   {\setlength{\temparraycolsep}{\arraycolsep}%
    \setlength{\arraycolsep}{1.5pt}%
        \left (%
    \begin{array}{ccc}%
                #1  & \blcligne{#2} & \blcligne{#3}
        \end{array}%
        \right )%
    \setlength{\arraycolsep}{\temparraycolsep}%
   }%
\newcommand{\PushLine}{\hbox{}\hfill\hbox{}}
\newlength{\DefiTest}
\newlength{\DefiHeightu}\newlength{\DefiHeightd}%
\newlength{\DefiDepthu}\newlength{\DefiDepthd}%
\newcommand{\Defi}[2]%
    {%
     \settoheight{\DefiHeightu}{${\displaystyle #1}$}%
     \settodepth{\DefiDepthu}{${\displaystyle #1}$}%
     \addtolength{\DefiHeightu}{\DefiDepthu}%
     \settoheight{\DefiHeightd}{${\displaystyle #2}$}%
     \settodepth{\DefiDepthd}{${\displaystyle #2}$}%
     \addtolength{\DefiHeightd}{\DefiDepthd}%
     \left\{#1%
     \rule[-\DefiDepthd]{\DefiTest}{\DefiHeightd}%
     \xmd\right|%
     \left.%
     \rule[-\DefiDepthu]{\DefiTest}{\DefiHeightu}%
      #2\right\}%
     }
\newlength{\ColoText}
\newlength{\ColoFigu}
\newlength{\TextFiguSpace}
\newlength{\parindenttemp} 
\newlength{\parskiptemp} 
\newlength{\fboxseptemp} 
\newcommand{\TFBoxing}{}
\newcommand{\TFVertAlig}{}
\newcommand{\LeftLarg}{}
\renewcommand{\LeftLarg}{.66}
\ifdraft\renewcommand{\TFBoxing}{\fbox}\fi
\newcommand{\TxtFg}[3]%
   {%
    \setlength{\ColoText}{#1\textwidth}%
    \setlength{\ColoFigu}{\textwidth}%
    \addtolength{\ColoFigu}{-\ColoText}%
    \addtolength{\ColoText}{-.5\TextFiguSpace}%
    \addtolength{\ColoFigu}{-.5\TextFiguSpace}%
    \ifdraft\setlength{\fboxsep}{0pt}\fi
    \noi
    \TFBoxing{%
       \begin{minipage}[\TFVertAlig]{\ColoText}%
          \setlength{\parindent}{\parindenttemp}%
          \setlength{\parskip}{\parskiptemp}%
          \par\vspace*{0mm}
             #2
       \end{minipage}%
             }%
    \hspace*{\TextFiguSpace}%
    \TFBoxing{%
       \begin{minipage}[\TFVertAlig]{\ColoFigu}%
          \par\vspace*{0mm}%
             #3%
       \end{minipage}%
             }%
    \ifdraft\setlength{\fboxsep}{\fboxseptemp}\fi
   }%
\newcommand{\TextFigu}[3][\LeftLarg]%
   {\renewcommand{\TFVertAlig}{t}\TxtFg{#1}{#2}{#3}}
\newcommand{\TextFiguC}[3][\LeftLarg]%
   {\renewcommand{\TFVertAlig}{c}\TxtFg{#1}{#2}{#3}}
\newcommand{\TextFiguX}[3][\LeftLarg]
   {%
    \setlength{\ColoText}{#1\textwidth}%
    \setlength{\ColoFigu}{\textwidth}%
    \addtolength{\ColoFigu}{-\ColoText}%
    \addtolength{\ColoText}{-.5\TextFiguSpace}%
    \addtolength{\ColoFigu}{-.5\TextFiguSpace}%
    \addtolength{\ColoFigu}{\ETAExtendedLineWidth}
    \ifdraft\setlength{\fboxsep}{0pt}\fi
    \noi
    \ifodd\value{page}%
       \TFBoxing{%
          \begin{minipage}[t]{\ColoText}%
             \RstBLS
             \setlength{\parindent}{\parindenttemp}%
             \setlength{\parskip}{\parskiptemp}%
             \par\vspace*{0mm}
                #2
          \end{minipage}%
                }%
       \hspace*{\TextFiguSpace}%
       \TFBoxing{%
          \begin{minipage}[t]{\ColoFigu}%
             \par\vspace*{0mm}%
                #3%
          \end{minipage}%
                }%
    \else
       \hspace*{-\ETAExtendedLineWidth}
       \TFBoxing{%
          \begin{minipage}[t]{\ColoFigu}%
             \par\vspace*{0mm}%
                #3%
          \end{minipage}%
                }%
       \hspace*{\TextFiguSpace}%
       \TFBoxing{%
          \begin{minipage}[t]{\ColoText}%
             \RstBLS
             \setlength{\parindent}{\parindenttemp}%
             \setlength{\parskip}{\parskiptemp}%
             \par\vspace*{0mm}
                #2
          \end{minipage}%
                }%
    \fi%
    \ifdraft\setlength{\fboxsep}{\fboxseptemp}\fi
   }
\newcommand{\Axio}[1]%
   {\pointn #1\hspace*{.1em}\jspointtiret\hspace*{.4em}\ignorespaces}

\newcommand{\path}{}

\newcommand{\jsTerval}[1]{[{#1}]} 

\newcommand{\pterval}{\jsTerval{p}}

\newcommand{\jsCard}[1]{{\|{#1}\|}}

\newcommand{\ExtnF}[1]%
   {\overset{{\scriptscriptstyle \pmb{\smile}}}{#1}}

\newcommand{\DiffF}[1]%
   {\overset{{\scriptscriptstyle \pmb{\lor}}}{#1}}

\newcommand{\LocaF}[1]%
   {\overset{{\scriptscriptstyle \leftrightarrow}}{#1}}

\newcommand{\jsDist}[2][{}]%
   {\operatorname{\mathbf{d}_{#1}}\left(#2\right)}

\renewcommand{\lim}{{\operatornamewithlimits{\mathsf{lim}}}}

\newcommand{\StruSA}[1]{\aut{#1}}

\newcommand{\x}{\! \times \!}

\newcommand{\SerSAnMon}[2]%
    {#1 \langle \! \langle  #2  \rangle \! \rangle }
\newcommand{\SerSAnMonD}[2]%
    {\left[#1\right] \langle \! \langle  #2  \rangle \! \rangle }
\newcommand{\SerMon}[1]%
    {\!\langle \! \langle  #1  \rangle \! \rangle }
\newcommand{\PolSAnMon}[2]%
    {{#1 \langle  #2 \rangle }}
\newcommand{\PolMon}[1]%
    {{\!\langle  #1 \rangle }}

\newsavebox{\LeftBraket}
\savebox{\LeftBraket}{\scalebox{0.7 1.2}{$<$}}
\newsavebox{\RightBraket}
\savebox{\RightBraket}{\scalebox{0.7 1.2}{$>$}}

\newcommand{\Rec}{\mathrm{Rec}\,}
\newcommand{\Rat}{\mathrm{Rat}\,}

\newcommand{\jsStar}[1]{{{#1}^{*}}}
\newcommand{\Ae}{\jsStar{A}}

\newcommand{\iotaK}{\iota_{\ShiftInd{K}}}

\newcommand{\compos}{\ccdot }

\newcommand{\phiikpsi}%
{{\varphi ^{-1}\! \compos        \iotaK \! \compos \! \psi }}
\newcommand{\phiiotpsi}[1]%
{{\varphi ^{-1}\! \compos        \iota _{\ShiftInd{#1}} \! \compos \! \psi }}
\newcommand{\phiintkpsi}[1]%
{{(#1\varphi ^{-1}\! \cap K) \psi }}

\newcommand{\jsgeq}{\geqslant }
\newcommand{\jsless}
   {\mathrel{\leqslant_{\!\!\!\!\scriptscriptstyle{/}}}}
\newcommand{\jsgrea}
   {\mathrel{\geqslant_{\!\!\!\!\scriptscriptstyle{\backslash}}}}

\newcommand{\lexiconeq}
   {\preccurlyeq_{\!\!\!\!\!\scalebox{1.8 1}{\scriptscriptstyle{\pmb{/}}}}}

\newcommand{\jsAutUn}[1]%
   {\mbox{$\left\langle \thinspace #1 \thinspace \right\rangle $}}
\newcommand{\aut}[1]{\jsAutUn{#1}} 
\newcommand{\auta}{\jsAutUn{Q,A,E,I,T}}

\newcommand{\ShiftInd}[1]{\raisebox{-0.3ex}{$\scriptstyle{#1}$}}

\newcommand{\act}{\mathbin{\boldsymbol{\cdot}}}

\newcommand{\actb}{\mathbin{\raisebox{0.2ex}%
                        {${\scriptscriptstyle \circ} $}}}
\newcommand{\ccdot}{\actb} 

\newlength{\vbh}\newlength{\vbd}\newlength{\vbt}%
\newcommand{\CompAuto}[1]%
    {%
     \settodepth{\vbd}{\mbox{$\displaystyle{#1\strut}$}}%
     \settoheight{\vbh}{\mbox{$\displaystyle{#1\strut}$}}%
     \setlength{\vbt}{\vbh}\addtolength{\vbt}{\vbd}%
     {}%
     \psline[linewidth=0.8pt]{c-c}(0,-.65\vbd)(0,.9\vbh)%
     \hspace*{0.7pt}%
     {#1}%
     \kern0.8pt%
     \psline[linewidth=0.8pt]{c-c}(0,-.65\vbd)(0,.9\vbh)%
     }%

\newcommand{\compA}{\CompAuto{\Ac}}

\newcommand{\bornedeuxlignes}[2]%
{\mbox{$
\begin{array}{c}{\scriptstyle #1}\\ {\scriptstyle #2} \end{array}
       $}}

\renewcommand{\path}[1]{\xrightarrow{\ #1 \ }} 
\newcommand{\pathaut}[2]{\underset{#2}{\path{#1}}}

\newcommand{\ExpDer}[2][a]%
    {\operatorname{\frac{\partial}{\partial \mbox{$#1$}}}#2}
\newcommand{\ExpDerP}[2][a]%
    {\operatorname{\frac{\partial}{\partial\mbox{$#1$}}}\left(#2\right)}
\newcommand{\ExpDerr}[2][a]%
    {\operatorname{\frac{\partial_{\mathrm{R}}}{\partial \mbox{$#1$}}}#2}
\newcommand{\ExpDerB}[2][a]%
   {\operatorname{\frac{\partial_\mathsf{b}}{\partial \mbox{$#1$}}}#2}
\newcommand{\ExpDerBP}[2][a]%
   {\operatorname{\frac{\partial_\mathsf{b}}{\partial \mbox{$#1$}}}\left(#2\right)}

\usepackage{graphicx}
\usepackage{caption}
\renewcommand{\CompAuto}[1]{{\pmb{|}{#1}\pmb{|}}}
\renewcommand{\x}{\xmd\! \times \!\xmd}
\newcommandy[o]{\N}{\Nmbb}
\newcommandy[o]{\Z}{\Zmbb}
\newcommand{\SL}{\Sc_{L}}
\newcommand{\Cond}[1]{\Cc_{#1}}
\newcommand{\ASNu}[1]{\val{\CompAuto{#1}}}
\renewcommand{\div}{\!\mid\!}
\newcommand{\notdiv}{\!\mathrel{\mid\!\!\!\!/}\!}
\newcommand{\MinusOne}[1]{{#1\!-\!\scalebox{.9}{1}}}
\newcommand{\PlusOne}[1]{#1\!\scalebox{.9}{+\!1}}
\newcommand{\bmo}{\MinusOne{b}}

\newcommand{\qmo}{\MinusOne{q}}

\newcommand{\npo}{\PlusOne{n}}
\renewcommand{\jsTerval}[1]{\{0,1,\ldots,\MinusOne{#1}\}} 
\newcommand{\Ab}{A_{b}}
\newcommand{\IdxdStar}[1]{{{#1}^{\!*}}}
\newcommand{\Abe}{\IdxdStar{\Ab}}
\newcommand{\Nd}{\N^{d}}
\newcommand{\brec}{$b$-recognisable\xspace}
\newcommand{\ifof}{if, and only if,\xspace}
\newcommand{\toscc}{Type~1 SCC\xspace}
\newcommand{\ttscc}{Type~2 SCC\xspace}
\newcommand{\up}{UP\xspace}
\newcommand{\upau}{UP-automaton\xspace}

\newcommand{\upcr}{UP-criterion\xspace}
\newcommand{\upsn}{UP-set of numbers\xspace}
\newcommand{\upssn}{UP-sets of numbers\xspace}
\newcommand{\uplbl}[1]{(\textsf{UP}-{#1})\xspace}
\newcommand{\upo}{\uplbl{1}}
\newcommand{\upt}{\uplbl{2}}
\newcommand{\uptr}{\uplbl{3}}
\newcommand{\upf}{\uplbl{4}}
\newcommand{\PSN}[2]{E_{#2}^{#1}}
\newcommand{\UPSN}[3]{E_{#2,#3}^{#1}}
\newcommand{\EpR}{\PSN{R}{p}}
\newcommand{\EpRm}{\UPSN{R}{p}{m}}
\newcommand{\Pasc}[2]{\Pc_{#2}^{#1}}
\newcommand{\PpR}{\Pasc{R}{p}}
\newcommand{\Pascp}[2]{{\Pc'}_{#2}^{#1}}
\newcommand{\PpRp}{\Pascp{R}{p}}
\newcommand{\ApR}{\PpR}
\newcommand{\PSA}[2]{\Bc_{#2}^{#1}}

\newcommand{\UPSA}[3]{\Bc_{#2,#3}^{#1}}
\newcommand{\GA}[1]{\Dc_{#1}}
\newcommand{\GPasc}[2]{\Gc_{#2}^{#1}}
\newcommand{\GpR}{\GPasc{R}{p}}
\newcommand{\sdp}{\xmd\! \rtimes \!\xmd}
\newcommand{\ZpZ}{\Z/p\Z}
\newcommand{\Zppsi}{\Z/p\Z\x\Z/\psi\Z}
\newcommand{\Zpsdpsi}{\Z/p\Z\sdp\Z/\psi\Z}
\newcommand{\RZpsi}{R\x\Z/\psi\Z}
\newcommand{\zz}{(0,0)}
\newcommand{\grg}{g} 
\newcommand{\gprod}{\xmd} 
\newcommand{\sdpG}[1]{G_{#1}} 
\newcommand{\sdpgp}{\sdpG{p}} 
\newcommand{\PermElF}[2]{{#1}_{#2}} 
\newcommand{\PermEl}[3]{\PermElF{#1}{#2}(#3)} 
\newcommand{\IndPer}[2]{\PermEl{\tau}{#1}{#2}} 
\newcommand{\ipst}[1]{\IndPer{(s,t)}{#1}} 
\newcommand{\ipsti}[1]{\PermElF{\tau}{(s,t)}^{-1}{#1}} 
\newcommand{\ipstf}{\PermElF{\tau}{(s,t)}} %
\newcommand{\IndPerT}[2]{\PermEl{\sigma}{#1}{#2}} 
\newcommand{\iptst}[1]{\IndPerT{(s,t)}{#1}} 
\newcommand{\iptstf}{\PermElF{\sigma}{(s,t)}} %
\newcommand{\Modu}[3]{#1\equiv#2\msp[#3]}
\newcommand{\mdp}[2]{\Modu{#1}{#2}{p}}
\newcommand{\mdk}[2]{\Modu{#1}{#2}{k}}
\newcommand{\mdd}[2]{\Modu{#1}{#2}{d}}

\newcommand{\equnm}[1]{(\ref{q.#1})\xspace}

\newcommand{\figur}[1]{Fig.\xmd\ref{f.#1}}

\newcommand{\propo}[1]{Proposition~\ref{p.#1}}

\newcommand{\theor}[1]{Theorem~\ref{t.#1}}
\newcommand{\secti}[1]{Sect.\xmd\ref{s.#1}}
\newcommand{\Rep}[1]{\langle#1\rangle}
\newcommandy[o]{\word}[1]{#1} 
\newcommandy[o]{\val}[1]{\overline{#1}} 
\newcommandy[o]{\wlen}[1]{|#1|}
\newcommandy[o]{\cod}[1]{\Rep{1}}
\newcommandy[o]{\ceil}[1]{\lceil #1 \rceil}
\newcommand{\quot}[1]{\widehat{#1}}
\newcommand{\gentwo}{10^{\uminus 1}}

\newcommand{\pascal}[3]{\Pasc{#1}{#2}}

\newcommand{\behavement}[1]{|#1|}
\newcommand\uminus{%
  \setbox0=\hbox{-}%
  \vcenter{%
    \hrule width\wd0 height \the\fontdimen8\textfont3%
    \vspace*{0.1em}
  }%
}
\newcommand{\abs}[1]{|#1|}

\newcommandy[i]{\ChgEdgeArrowWidth}[1]{}
\newcommandy[i]{\RstEdgeArrowWidth}{}

\def\curlang{en}

\SetVCDirectory{./}

\ifdraft\ShowFrame\ShowGrid\fi
\TinyPicture

\begin{document}
\maketitle
\thispagestyle{plain}
\begin{abstract}
  It is decidable if a set of numbers, whose representation in a base~$b$ is a
    regular language, is ultimately periodic.
  This was established by Honkala in 1986.

  We give here a structural description of minimal automata that accept an
    ultimately periodic set of numbers.
  We then show that it can be verified in linear time if a given minimal
  automaton meets this description.

  This yields a $O(n\xmd $log$(n))$ procedure for deciding whether a
    general deterministic automaton accepts an ultimately periodic set of
    numbers.
\end{abstract}

\section{Introduction}

Given a fixed positive integer~$b$, called the \emph{base}, every
positive integer~$n$ is represented (in base~$b$) by a \emph{word}
over the digit alphabet
$\msp\Ab\nlb=\nlb\{0,1,\ldots,\bmo\}\msp$
which does not start with a~$0$.
Hence, \emph{sets} of numbers are represented by \emph{languages}
of~$\Abe$.
Depending on the base, a given set of integers may be represented by
a simple or complex language:
the set of powers of~$2$ is represented by the rational
language~$10^*$ in base~$2$; whereas in base~$3$, it can only be represented
by a context-sensitive
language, much harder to describe.

A set of numbers is said to be \emph{\brec}if it is
represented by a recognisable, or rational, or regular, language
over~$\Abe$.
On the other hand, a set of numbers is \emph{recognisable} if it is,
via the identification of~$\N$ with~$a^{*}$
($n \leftrightarrow a^{n}$),
a recognisable, or rational, or regular, language of the free
monoid~$a^{*}$.
A set of numbers is recognisable if, and only if it is
\emph{ultimately periodic} (\up) and we use the latter terminology in
the sequel
as it is both meaningful and more distinguishable from \brec.
It is common knowledge that every \upsn is \brec for
every~$b$, and the above example shows that a~{$b$-recognisable} set for
some~$b$ is not necessarily UP, nor~$c$-recognisable for all~$c$.
It is an exercice to show that if~$b$ and~$c$ are \emph{multiplicatively
dependent} integers (that is, there exist~integers~$k$ and~$l$ such that
$b^{k}=c^{l}$), then every \brec set is a~$c$-recognisable set as
well (\cf \cite{FrouSaka10hb} for instance).
A converse of these two properties is the theorem of Cobham:
\emph{a set of numbers which is both~$b$- and~$c$-recognisable, for
multiplicatively independent~$b$ and~$c$, is UP},
established in~1969~\cite{Cobh69}, a strong and deep result whose
proof is difficult (\cf\cite{BruyEtAl94}).

After Cobham's theorem, the next natural (and last) question left
open on~\brec sets of numbers was the decidability of ultimate periodicity.
It was positively solved in~$1986$:

\begin{theorem}[Honkala~\cite{Honk86}]
\label{t.hon}%
It is decidable whether an automaton over~$\Abe$ accepts a \upsn.
\end{theorem}

The complexity of the decision procedure was not
an issue in the original work. 
Neither were the properties or the structure of automata accepting
\upsn.
Given an automaton~$\Ac$ over~$\Abe$, bounds are computed on the parameters 
of a potential \upsn accepted by~$\Ac$.
The property is then decidable as it is possible to enumerate all
automata that accept sets with smaller parameters and check whether any of them is equivalent to~$\Ac$.

As explained below, subsequent works on automata and number
representations brought some answers regarding the complexity
of the decision procedure, explicitly or implicitly.
The  present paper addresses specifically this problem and yields the
following statement.

\begin{theorem}
\label{t.com-plx}%
It is decidable \emph{in linear time} whether a minimal DFA~$\Ac$
over~$\Abe$ accepts a \upsn.
\end{theorem}

As it is often the case, this complexity result is obtained as the consequence
of a structural characterisation.
Indeed, we describe here a set of structural properties for an
automaton: the shape of its strongly connected components (SCC's) and
that of
its graph of SCC's, that we gather under the name of \upcr.
\theor{com-plx} then splits into two results:

\begin{theorem}
\label{t.upc-cha}%
A minimal DFA~$\Ac$
over~$\Abe$ accepts a \upsn
if, and only if, it
satisfies the \upcr.
\end{theorem}

\begin{theorem}
\label{t.upc-lin}%
It is decidable {in linear time} whether a minimal DFA~$\Ac$
over~$\Abe$ satisfies the \upcr.
\end{theorem}

As for Cobham's theorem (\cf~\cite{BruyEtAl94,DuraFRigo11}), new
insights on the problem tackled here are obtained when stating it in
a higher dimensional space.
Let~$\N^{d}$ be the additive monoid of~$d$-tuples of integers.
Every~$d$-tuple of integers may be represented in base~$b$ by a
$d$-tuple of words of~$\Abe$ of \emph{the same length}, as shorter
words can be padded by~$0$'s without changing the corresponding value.
Such~$d$-tuples can be read by (finite) automata
over~$({\Ab}^{\! d})^{*}$ --- automata reading on~$d$ synchronised
tapes ---
and a subset of~$\N^{d}$ is \brec if the set of the
$b$-representations of its elements is accepted by such an automaton.

On the other hand, recognisable and rational sets of~$\Nd$ are
defined in the classical way but they do not coincide
as~$\Nd$ \emph{is not a free monoid}.
A subset of~$\Nd$ is \emph{recognisable} if is saturated by a
congruence of finite index, and the family of recognisable sets is
denoted by~$\Rec\Nd$.
A subset of~$\Nd$ is \emph{rational} if is denoted by a rational
expression, and the family of rational sets is
denoted by~$\Rat\Nd$.
Rational sets of~$\Nd$ have been characterised by Ginsburg and
Spanier as sets definable in the \emph{Presburger arithmetic}
$\StruSA{\N,+}$ (\cite{GinsSpan66}), hence the name \emph{Presburger
definable} that is most often used in the literature.

It is also common knowledge that every rational set of~$\Nd$ is \brec for
every~$b$, and the example in dimension~$1$ is enough to show that a \brec set
is not necessarily rational.
The generalisation of Cobham's theorem:
\emph{a subset of~$\Nd$ which is both~$b$- and~$c$-recognisable, for
multiplicatively independent~$b$ and~$c$, is rational}, is due
to Semenov (\cf~\cite{BruyEtAl94,DuraFRigo11}).
The generalisation of Honkala's theorem went as smoothly.

\begin{theorem}[Muchnik~\cite{Much03}]
\label{t.muc}%
It is decidable whether a \brec subset of~$\Nd$ is rational.
\end{theorem}

\begin{theorem}[Leroux~\cite{Lero05}]
\label{t.ler}%
It is decidable \emph{in polynomial time} whether
a minimal DFA~$\Ac$
over~$({\Ab}^{\! d})^{*}$ accepts a rational subset
of~$\Nd$.
\end{theorem}

The algorithm underlying \theor{muc} is triply exponential whereas
the one described in~\cite{Lero05}, based on sophisticated geometric
constructions, is quadratic --- an impressive improvement --- but not
easy to explain.

There exists another way to devise a proof for Honkala's theorem which yields
another extension.
In \cite{GinsSpan66}, Ginsburg and Spanier also proved that there exists a formula
in Presburger arithmetic deciding whether a given subset of~$\Nd$ is recognisable.
In dimension 1, it means that being a \upsn is expressible in Presburger 
arithmetic.
%
%
%
In \cite{AlloEtAL09}, it was then noted that since addition in base~$p$
is realised by a finite automaton, every Presburger formula is realised by
a finite automaton as well. Hence a decision procedure that establishes
\theor{hon}.

Generalisation of base~$p$ by non-standard numeration systems then gives an
extension of \theor{hon}, best expressed in terms of abstract numeration
systems.
Given a totally ordered alphabet~$A$, any \emph{rational} language~$L$
of~$\Ae$ defines an \emph{abstract numeration system}
(ANS)~$\SL$ in which the integer~$n$ is represented by the~$\npo$-th
word of~$L$ in the radix ordering of~$\Ae$
(\cf~\cite{LecoRigo10hb}).
A set of integers whose representations in the ANS~$\SL$ form a
rational language is called~$\SL$-recognisable and it is known that
every \upsn is~$\SL$-recognisable for every ANS~$\SL$
(\cite{LecoRigo10hb}).
The next statement then follows.

\begin{theorem}
\label{t.all-eta}%
If~$\SL$ is an abstract numeration system in which addition is
realised by a finite automaton, then
it is decidable whether a~$\SL$-recognisable set of numbers is \up.
\end{theorem}

For instance, \theor{all-eta} implies that ultimate periodicity is decidable
for sets of numbers represented by rational sets in a Pisot base
system \cite{Frou92}.
The algorithm underlying \theor{all-eta} is exponential
(if the
set of numbers is given by
a DFA)
and thus
(much) less efficient than Leroux's constructions for integer base
systems.
On the other hand, it applies to a much larger family of
numeration systems.
All this was mentioned for the sake of completeness, and  the present
paper does not follow this pattern.

\medskip

\theor{ler}, restricted to dimension~$1$, readily yields a quadratic procedure
for Honkala's theorem. The improvement from quadratic to quasilinear complexity
achieved in this article is not a natural simplification of Leroux's construction
for the case of dimension~1. Although the UP-criterion bears similarities with
some features of Leroux's construction, it is not derived from~\cite{Lero05}, nor is the
proof of quasilinear complexity.

\medskip

The paper is organised as follows.
In \rsection{pas-cal}, we treat the special case of determining whether a given minimal
group automaton accepts an ultimately periodic set of numbers. We
describe canonical automata, which we call Pascal automata, that accept such sets.
We then show how to decide in linear time whether a given
minimal group automaton is the quotient of some Pascal automaton.

\rsection{upc} introduces the UP-criterion and sketches both its completeness
and correctness. An automaton satisfying the UP-criterion is a directed
acyclic graph (DAG) 'ending' with at most two layers of non-trivial strongly connected
components (SCC's). If the root is seen at the top, the upper (non-trivial) SCC's are circuits
of 0's and the lower ones are quotients of Pascal automata.
It is easy, and of linear complexity to verify that an automaton has this
overall structure.
This criterium is sketched in \rfigure{upc}.
\newline

\TinyPicture
\begin{figure}[ht!]
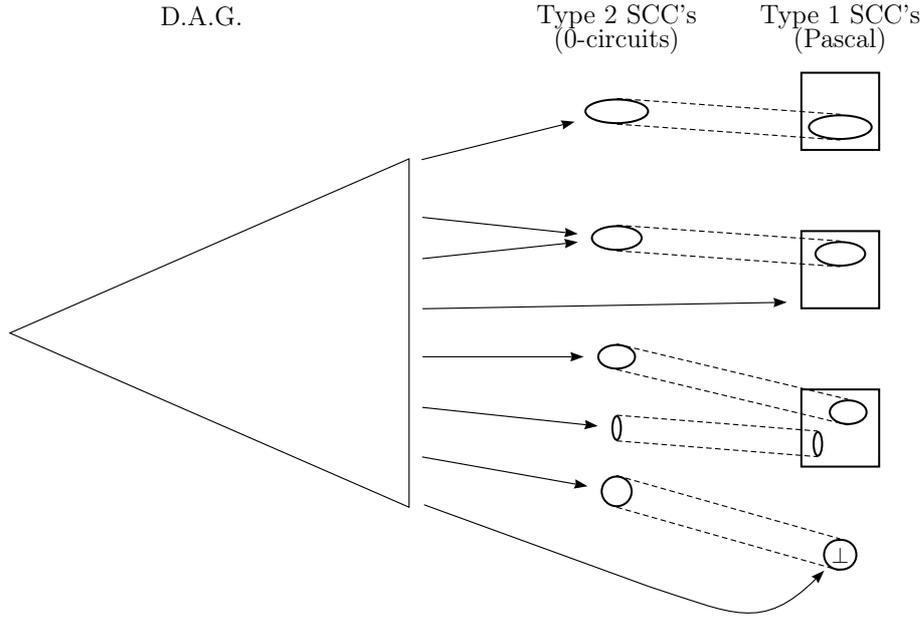

  \centering
  \VCCall{up_criterion.vcsg}
  \caption{A schematic representation of the \upcr}
  \label{f.upc}
\end{figure}
\MediumPicture


\section{The Pascal automaton}
\label{s.pas-cal}%

\subsection{Preliminaries}
\label{s.pas-pre}%

\subsubsection{On automata}
\label{s.pas-pre-aut}%
We consider only finite deterministic finite automata, denoted
  by~$\msp\Ac=\aut{Q,A,\delta,i,T}\msp$, where~$Q$ is the set of
  \emph{states},~$i$ the \emph{initial state} and~$T$ the set
  of \emph{final states};~$A$ is the \emph{alphabet},~$A^*$ is
  the \emph{free monoid} generated by~$A$ and the \emph{empty word} is denoted
  by~$\epsilon$;~${\delta:Q\times A \rightarrow Q}$ is the \emph{transition
  function}.

As usual,~$\delta$ is extended to a function~$Q\x A^*\rightarrow Q$
    by~$\delta(q,\epsilon)=q$ \linebreak
    and~$\delta(q,ua)\nlb=\nlb\delta(\delta(q,u),a)$;
    and~$\delta(q,u)$ will also be denoted by~$q\cdot u$.
When~$\delta$ is a total function,~$\mathcal A$ is said to be \emph{complete}.
In the sequel, we only consider automata that are \emph{accessible}, \ie in which
every state is reachable from~$i$.

A word~$u$ of~$A^*$ is \emph{accepted}
  by~$\Ac$ if~$\msp i \cdot u \msp$ is in~$T$.
The set of words accepted by~$\Ac$ is called \emph{the} language of~$\Ac$,
  and is denoted by~$\behavement{\Ac}$.
%

Let~$~\Ac=\aut{Q,A,\delta,i,T}$ and~$~\Bc=\aut{R,A,\eta,j,S}$ be two
deterministic automata.
A map~$\phi:Q\rightarrow R$ is an \emph{automaton morphism},
written~${\phi:\Ac\rightarrow\Bc}$
if~$\phi(i)=j$,~$\phi(T)\subseteq S$,
and for all~$q$ in~$Q$ and~$a$ in~$A$, such that~$\delta(q,a)$ is defined,
then~$\eta(\phi(q),a)$ is defined, and~$\phi(\delta(q,a))=\eta(\phi(q),a)$.
We call~$\phi$ a \emph{covering} if the following two conditions hold: 
  i)~$\phi(T)=S$ and 
  ii) for all~$q$ in~$Q$ and~$a$ in~$A$, if~$\eta(\phi(q),a)$ is defined,
  then so is~$\delta(q,a)$. 
In this case,~$\behavement{\Ac}=\behavement{\Bc}$, and~$\Bc$ is called
  \emph{a quotient} of~$\Ac$.
Note that if~$\Ac$ is complete, every morphism satisfies (ii).

Every complete deterministic automaton~$\Ac$ has a \emph{minimal quotient} which is 
  \emph{the minimal automaton} accepting~$\behavement{\Ac}$.
This automaton is unique up to isomorphism and can be computed from~$\Ac$
  in~$O(n\xmd log(n))$ time, where~$n$ is the number of
  states of~$\Ac$ (\cf~\cite{AhoHopcUllm74}).

Given a deterministic automaton~$\Ac$, every word~$u$ induces an appli-\linebreak
  cation~(${q\xmd\mapsto\xmd q\cdot u}$) over the state set.
These applications form a finite monoid, called the
  \emph{transition monoid} of~$\Ac$.
When this monoid happens to be a group (meaning that the action
  of every letter is a permutation over the states),~$\Ac$ is
  called a \emph{group automaton}.

\subsubsection{On numbers}
\label{s.pas-pre-num}%
The base~$b$ is fixed throughout the paper (it will be a
  \emph{parameter} of the algorithms, not an input) and so
  is the digit alphabet~$\Ab$.
As a consequence, the number of transitions of any deterministic 
  automaton over~$\Abe$ is
  linear in its number of states.
Verifying that an automaton is deterministic (resp. a group
  automaton) can then be done in linear time.

For our purpose, it is far more convenient to write the integers
  \emph{least significant digits first} (LSDF), and to keep the
  automata reading \emph{from left to right} (as in Leroux's work~\cite{Lero05}).
The \emph{value} of a word~$u=\word{a_{0}a_{1}\cdots a_n}$ of~$\Abe$,
  denoted by~$\val{u}$, is then~$\msp\val{u}\nlb=\nlb\sum_{i=0}^n(a_ib^i)\msp$
  and may be obtained by the recursive formula:
\begin{equation}
  \val{u\xmd a} = \val{u} + a \xmd b^{\wlen{u}}
  \label{q.val}
\end{equation}

Conversely, every integer~$n$ has a unique canonical representation
  in base~$b$ that does not \emph{end} with~$0$, and is denoted by~$\Rep{n}$.
A word of~$\Abe$ has value~$n$ if, and only if, it is of the
  form~$\Rep{n}\xmd 0^{k}$.

By abuse of language, we may talk about the \emph{set of numbers}
  accepted by an automaton.
An integer~$n$ is then accepted if there exists a word of value~$n$
  accepted by the automaton.

A set~$E\subseteq\N$ is \emph{periodic}, of \emph{period}~$q$, if
  there exists~$S\subseteq\{0,1,\ldots,\qmo\}$ such that
  $\msp E=\Defi{n\in\N}{\ext r\in S \quantsmsp n\equiv r\xmd [q]}$.
Any periodic set~$E$ has a \emph{smallest} period~$p$
  and a corresponding set of \emph{residues}~$R$:
  the set~$E$ is then denoted by~$\EpR$.
The set of numbers in~$\EpR$ and larger than an integer~$m$ is
  denoted by~$\EpRm$.

\subsection{Definition of a Pascal automaton}
\label{s.def-pas}%

We begin with the construction of an automaton~$\ApR$ that accepts the
  set~$\EpR$, in the case where
\vspace*{-1em}
\begin{center}
  \emph{$p$ is coprime with~$b$}.
\end{center}
We call any such automaton a \emph{Pascal automaton}.%
    \footnote{
      As early as~1654, Pascal describes a computing process that
      generalises the casting out nines and that determines if an
      integer~$n$, written \emph{in any base~$b$}, is divisible by an
      integer~$p$ (see~\cite[Prologue]{Saka09}).%
    }
If~$p$ is coprime with~$b$, there exists a (smallest positive)
  integer~$\psi$ such that:
\begin{equation}
	\mdp{b^{\psi}}{1}
	\ee\text{and thus}\ee
	\fa x\in\N\quantsp
	\mdp{b^{x}}{b^{x\mod \psi}}
	\eqpnt
	\notag
\end{equation}
Therefore, from Equation~(\ref{q.val}), knowing~$\val{u}\mod p$
  and~$\wlen{u}\mod \psi$ is enough to compute~$\val{u\xmd a}\mod p$.

\noindent Hence the definition of
  $\msp\ApR=\aut{\Zppsi,\Ab,\eta,\zz,\RZpsi}\msp$, where
\begin{equation}
	\fa(s,t)\in\Zppsi\quantvrg
	\fa a\in\Ab\quantsp
	\msp \eta((s,t),a) = (s,t)\cdot a = (s+ a\xmd b^t,t+1)
	\label{q.pas-tra-1}
\end{equation}

By induction on~$\wlen{u}$, it follows that
  $\msp \zz\cdot u = (\val{u}\mod p,\wlen{u}\mod\psi)\msp$
  for every~$u$ in~$\Abe$ and consequently that~$\EpR$ is the set of number 
  accepted by~$\ApR$.

\begin{example}
  \label{e.aut_2mod3_2}%
  \figur{aut_2mod3_2} shows~$\Pasc{2}{3}$, the Pascal automaton
    accepting integers written in binary and congruent
    to~$2$ modulo~$3$.
    For clarity, the labels are omitted;
      transitions labelled by~$1$ are drawn with thick lines and those labelled
      by~$0$ with thin lines.
\end{example}

\begin{figure}[ht!]
  \centering
  \VCCall{aut_2mod3_2}
  \caption{The Pascal automaton~$\Pasc{2}{3}$}
  \label{f.aut_2mod3_2}
\end{figure}

In fact, this construction does not require that~$p$ and~$R$ be canonical.
For arbitrary~$p$ (still prime with~$b$) and~$R$, we call the automaton
  constructed in this way a \emph{generalised Pascal automaton} and denote it
  by~$\GpR$.

\subsection{Recognition of quotients of Pascal automata}
\label{s.rec-pas}%

The tricky part of achieving a linear complexity for \theor{upc-lin} is
  contained in the following statement:

\begin{theorem}\ltheorem{pas-quo-lin}
  It is decidable in linear time whether a minimal DFA~$\Ac$ over~$\Ab$ is the
  quotient of a Pascal automaton.
\end{theorem}

\paragraph*{Simplifications}
\label{s.pas-sim}%

Since~$\PpR$ is a group automaton, all its quotients are group automata.


The permutation on~$\Zppsi$ realised by~$0^{(\psi-1)}$ is the inverse
  of the one realised by~$0$ and we call it the
  action of the ``digit''~$0^{\uminus1}$.
Let~$\grg$ be a new letter whose action on~$\Zppsi$ is the one
  of~$\gentwo$.
It follows from~\equnm{pas-tra-1} that for every~$a$ in~$\Ab$
    --- where~$a$ is understood both
    as a \emph{digit} and as a \emph{number} ---
  the action of~$a$ on~$\Zppsi$ (in~$\PpR$) is equal to the one
  of~$g^{a}0$.
The same relation holds in any group automaton~$\Ac$ over~$\Abe$ that
  is a quotient of a Pascal automaton, and this condition is
  tested in linear time.

Let~$B=\{0,\grg\}$ be a new alphabet.
Any group automaton~$\Ac=\aut{Q,\Ab,\delta,i,T}$ may be transformed into
  an automaton~$\Ac'=\aut{Q,B,\delta',i,T}$ where, for every~$q$ in~$Q$,
  $\msp\delta'(q,0)=\delta(q,0)\msp$
  and~$\msp\delta'(q,\grg)=\delta(q,\gentwo)\msp$.

\figur{aut_2mod3_2_v} shows~$\Pascp{2}{3}$ where transitions
  labelled by~$0$ are drawn with thin lines and those labelled
  by~$g$ with double lines.%
    \footnote{%
      The transformation highlights that the transition monoid
      of~$\PpR$ (and thus of~$\PpRp$) is the \emph{semi-direct
      product}~$\Zpsdpsi$.%
    }

\begin{figure}[ht!]
  \centering
  \VCCall{aut_2mod3_2_v}
  \caption{The modified Pascal automaton~$\Pascp{2}{3}$}
  \label{f.aut_2mod3_2_v}
\end{figure}
%
\paragraph*{Analysis: computation of the parameters}
\label{s.pas-com}%

From now on, and for the rest of the
  section,~$\Ac=\aut{Q,\Ab,\delta,i,T}$ is a group automaton which has
  been consistently transformed into an
  automaton~$\Ac'=\aut{Q,B,\delta',i,T}$.
If~$\Ac'$ is a quotient of a Pascal automaton~$\PpRp$, then the
  parameters~$p$ and~$R$ may be computed (or~`read') in~$\Ac'$;
  this is the consequence of the following
  statement.

\begin{proposition}\label{p.can-par}%
  Let~$\msp\phi\colon\PpRp\rightarrow\Ac'\msp$
    be a covering.
  Then, for every~$(x,y)$ \linebreak and~$(x',y')$ in~$\Zppsi$,
    if~$x\neq x'$ and~$\phi(x,y)=\phi(x',y')$, then~$y\not=y'$.
\end{proposition}

\begin{proof}
Ab absurdo.
Since
$\msp(x,y)\pathaut{0^{-y}\grg^{-x}}{\PpRp}(0,0)\msp$,
$\phi(x,y)=\phi(x',y)$ implies~$\phi(0,0)=\phi(z,0)$ with~$z=x'-x$.
Such~$z$ may be chosen minimal, and the image by~$\phi$ of the
$\grg$-circuit containing~$(0,0)$ in~$\PpRp$ (of length~$p$) is a
$\grg$-circuit in~$\Ac'$ of length~$z$.

Since the image, and the inverse image, by a covering of a
final state is final, if~$r$ is in~$R$, that is, if~$(r,0)$ is final,
so is~$(r+z,0)$, and the pair~$p,R$ is not canonical.
\end{proof}

\begin{corollary} \label{c.can-par}%
  If~$\Ac'=\aut{Q,B,\delta',i,T}$ is a quotient of a modified Pascal
    automaton~$\PpRp$, then~$p$ is the length of the~$\grg$-circuit
    in~$\Ac'$ which contains~$i$ and~$R=\Defi{r}{i\cdot \grg^{r}\in T}$.
\end{corollary}

Next, if~$\Ac'$ is a quotient of a (modified) Pascal
  automaton~$\PpRp$, the equivalence class of the initial
  state of~$\PpRp$ may be `read' as well in~$\Ac'$ as
  the intersection of the~$0$-circuit and the~$\grg$-circuit
  around the initial state of~$\Ac'$.
More precisely, and since
  $\msp(0,0)\pathaut{\grg^{s}}{\PpRp}(s,0)\pathaut{0^{t}}{\PpRp}(s,t)\msp$,
  the following holds.

\begin{proposition}\label{p.can-par2}%
  Let~$\msp\phi\colon\PpRp\rightarrow\Ac'\msp$
    be a covering.
  For all~$s$ in~$\Z/p\Z$ and $t$ in~${\Z/\psi\Z}$,~${\phi(s,t)=\phi(0,0)}$ \ifof
    $\msp i\cdot \grg^{s}=i\cdot0^{\uminus t}\msp$.
\end{proposition}

From this proposition follows that, given~$\Ac'$, it is easy to
  compute the class of~$(0,0)$ modulo~$\phi$ if~$\Ac'$ is indeed a
  quotient of a (modified) Pascal automaton by~$\phi$.
Starting from~$i$, one first marks the states on the~$\grg$-circuit~$C$.
  Then, starting from~$i$ again,  one follows
  the~$0^{\uminus1}$-transitions: the \emph{first time}~$C$ is
  crossed yields~$t$.
This parameter is \emph{characteristic} of~$\phi$, as explained now.

Let~$(s,t)$ be an element of the semidirect
  product~$\sdpgp=\Zpsdpsi$ and~$\ipstf$ the permutation on~$\sdpgp$
  induced by the multiplication \emph{on the left} by~$(s,t)$:
\begin{equation}
	\ipst{(x,y)} = (s,t)\xmd(x,y)= (x\xmd b^{t}+ s, y + t)
	\eqpnt
	\label{q.pas-quo-1}
\end{equation}
The same element~$(s,t)$ defines a permutation~$\iptstf$ on~$\ZpZ$ as
  well:
\begin{equation}
	\fa x\in\ZpZ\quantsp
	\iptst{x} = x\xmd b^{t}+ s
	\eqpnt
	\label{q.pas-quo-2}
\end{equation}

Given a permutation~$\sigma$ over a set~$S$, \emph{the orbit}
of an element~$s$ of~$S$ under~$\sigma$ is the set~$\set{\sigma^i(s)~|~i\in\N}$.
\emph{An orbit} of~$\sigma$ is one of these sets.

\begin{proposition}
\label{p.can-cla}%
  Let~$\msp\phi\colon\PpRp\rightarrow\Ac'\msp$
    be a covering
    and let~$(s,t)$ be the state~$\phi$-equivalent to~$\zz$ with the
    smallest second component.
  Then, every $\phi$-class is an orbit
    of~$\ipstf$ (in~$\sdpgp$) and~$R$ is an union of orbits of~$\iptstf$ (in~$\Z/p\Z$).
\end{proposition}
\begin{proof}  Since~$\phi$ is a covering and~$\ipstf$ multiplies on the left 
  by~$(s,t)$ (which is~$\phi$-equivalent to~(0,0)), it follows 
  that~$\ipstf$ is stable on every~$\phi$-equivalence class.
  This already proves that~$R$ is a union of orbits of~$\iptstf$.

  Moreover, one can reduce the case where there are two elements~$h$ and~$h'$
  in the same~$\phi$-equivalence class but
  in a different {$\ipstf$-orbit}, to the case where one of them is~$(0,0)$,
  by multiplying both by~$h^{\uminus1}$ on the right.
  
  \smallskip


    We denote by~$C$ the~$\phi$-equivalence class of~$(0,0)$.
    The action of~$\ipstf$ on the second component is simply to add~$t$.
    It follows that there cannot be two vertices~$(x_1,y_1)$ and~$(x_2,y_2)$
    in~$C$ such that~$\abs{y_1-y_2}<t$, since applying~$\ipstf$ enough time
    to both would yield a vertex in~$C$ whose second component is smaller
    than~$t$.

    For all~$j$,~$\ipstf^j(0,0)$ is equal to~$(z_j,j\xmd t)$ for some~$z_j$,
    and is in~$C$ (since~$\ipstf^j$ is stable over~$C$).
    There cannot be any other state~$(x,y)$ in~$C$, since otherwise there would exist
    some~$j$ such that~$j\xmd t\leq y < (j+1)\xmd t$, and
    then~$\abs{y - j\xmd t}<t$.

    Hence~$C$ is the orbit of (0,0) for~$\ipstf$.
\end{proof}
\paragraph*{Synthesis: verification that a given automaton is a
quotient of a Pascal automaton}
\label{s.pas-ver}%

Given~$\Ac'=\aut{Q,B,\delta',i,T}$, let~$p$,~$R$ and~$(s,t)$ computed
as explained above.
It is easily checked that~$R$ is an union of orbits of~$\iptstf$ and
that~$\jsCard{Q}=p\xmd t$.
The last step is the verification that~$\Ac'$ is indeed (isomorphic
to) the quotient of~$\PpRp$ by the morphism~$\phi$ defined by~$(s,t)$.

A corollary of \propo{can-cla} (and of the multiplication law
in~$\sdpgp$) is that every class modulo~$\phi$ contains one, and
exactly one, element whose second component is smaller than~$t$.
From this observation follows that the multiplication by the
generators~$0=(0,1)$ and~$\grg=(1,0)$ in the quotient of~$\PpRp$
by~$\phi$ may be described on the set of representatives\\
\PushLine
$\msp Q_{\phi}=\Defi{(x,z)}{x\in\Z/p\Z, z\in\Z/t\Z}\msp$
\PushLine\\
(beware that~$z$ is in~$\Z/t\Z$
and not in~$\Z/\psi\Z$) by the following formulas:
\begin{align}
	\fa (x,z)\in Q_{\phi}\quantsp
	\eee & \e \notag\\
	(x,z)\act 0 = (x,z)\gprod (0,1) &=
{\begin{cases}
      \; (x, z+1) & \text{\e if  \e $z < t-1$} \\[.8ex]
      \;  \ipsti{(x,z+1)}= (\frac{x-s}{b^{t}},0)& \text{\e if  \e $z = t-1$} \\
\end{cases} }
	\notag
	\\[1.5ex]
	(x,z)\act \grg = (x,z)\gprod (1,0) &=
	(x+b^{z},z)
	\eqpnt
	\notag
\end{align}

Hence~$\Ac'$ is the quotient of~$\PpRp$ by~$\phi$ if one can mark~$Q$
according to these rules, starting from~$i$ with the mark~$(0,0)$,
without conflicts and in such a way that two distinct states have
distincts marks.
Such a marking is realised by a simple traversal of~$\Ac'$, thus in linear
time, and this concludes the proof of \rtheorem{pas-quo-lin}.

\begin{remark}\lremark{initial_shift}
  \rtheorem{pas-quo-lin} states that one can decide in linear time whether
    a \emph{given} automaton~$\Ac$ is a quotient of a Pascal automaton,
    and in particular~$\Ac$ has a fixed initial state
    that plays a crucial role in the verification process.

  The following proposition shows that the property (being a quotient of a
    Pascal automaton) is actually independent of the state chosen to be initial.
  If it holds for~$\Ac$, it also holds for any automaton derived from~$\Ac$ by
    changing the initial state.
  This is a general property that will be used in the
    general verification process described in the next section.
\end{remark}

\begin{proposition}\lproposition{initial_shift}
  If an automaton~$\Ac=\aut{Q,\Ab,\delta,i,T}$ is the quotient of~$\PpR$,
  then for every state~$q$ in~$Q$,~$\Ac_q$=$\aut{Q,\Ab,\delta,q,T}$ is the
  quotient of~$\Pasc{S}{p}$ for some set~$S$.
\end{proposition}
\begin{proof}
  Since the morphism associated with a quotient does not depend on the initial
  state, it is enough to prove that changing the initial state of a Pascal 
  automaton yield   another Pascal automaton with the same period.

  It is then easy to verify that, if the new initial state is~$(s,t)$, the
  new automaton is equal to~$\Pasc{S}{p}$ 
    where~$S = \set{\frac{r-s}{b^t}~|~r\in R}$,
    the state~$(x,y)$ of ~$\Pasc{S}{p}$ corresponding to the 
    state~$(s+xb^t,y+t)$ of~$\Pasc{R}{p}$.
\end{proof}


\section{The \upcr}
\label{s.upc}%

Let
$\msp\Ac=\auta\msp$
be an automaton,
$\sigma$ the strong connectivity equivalence relation on~$Q$,
and~$\gamma$ the surjective map from~$Q$ onto~$Q/\sigma$.
The \emph{condensation}~$\Cond{\Ac}$ of~$\Ac$ is 
the directed acyclic graph with loops~$(V,E)$
such that~$V$ is the image of~$Q$ by~$\gamma$;
and the edge~$(x,y)$ is in~$E$ if there exists a transition~$\msp q\pathx{a}s \msp$ in~$\Ac$,
for some~$q$ in~$\gamma^{\uminus1}(x)$,~$s$ in~$\gamma^{\uminus1}(y)$ and~$a$ in~$A$.
The condensation of~$\Ac$ can be computed in linear time by Tarjan's
algorithm (\cf~\cite{CormEtAl09}).

We say that an SCC~$C$ of an automaton~$\mathcal{A}$ is \emph{embeddable} in another
  SCC~$D$ of~$\Ac$ if there exists an injective function~${f:C\rightarrow D}$ such that,
    for all~$q$ in~$C$ and~$a$ in~$A$: if~$\msp q\cdot a \msp$ is in~$C$ then $\msp f(q\cdot a) = (f(q)\cdot a)\msp$,
    and if~$q\cdot a$ is not in~$C$, then~$f(q)\cdot a = q \cdot a$.
\begin{definition}[The \upcr]
  Let~$\Ac$ be a complete deterministic automaton and~$\Cond{\Ac}$ its condensation.
  We say that~$\Ac$ satisfies the \upcr (or equivalently that~$\Ac$ is a \upau) if the following
  five conditions hold.
  \begin{description}
    \item[UP-0] The successor by~$0$ of a final (resp.
      non-final) state of~$\Ac$ is final (resp.  non-final).

    \item[UP-1] Every non-trivial SCC of~$\Ac$ that contains an internal transition
    labelled by a digit different from~0 is mapped by~$\gamma$ to a leaf of~$\Cond{\Ac}$.

    Such and SCC is called a \toscc.

    \item[UP-2]
    Every non-trivial SCC of~$\Ac$ which is not of Type~1:

      i) is a simple circuit labelled by 0 (or 0-circuit);

      ii) is mapped by~$\gamma$ to a vertex of~$\Cond{\Ac}$ which has a unique
          successor, and this successor is a leaf.

    \noindent Such an SCC is called a Type~2 SCC.

    \item[UP-3]
    Every \toscc is the quotient of a
    Pascal automaton~$\Pasc{R}{p}$, for some~$R$ and~$p$.

    \item[UP-4]
      Every \ttscc is embeddable in the unique \toscc associated with it by~(UP-2).
  \end{description}
  %
\end{definition}

It should be noted that~(UP-0) is not a specific condition, it is more of
a precondition (hence its numbering~0) to ensure that either all representations
of an integer are accepted, or none them are.
Moreover,~(UP-1) and~(UP-2) (together with the completeness of~$\Ac$) imply the
converse of~(UP-1), namely that every SCC mapped by~$\gamma$ to a leaf 
of~$\Cond{\Ac}$ is a \toscc.

The schematic representation of the \upcr at \figur{upc} allows to 
review items~1 to~4.
There are two levels of SCC's in the condensation; squares and ovals.
Squares are the \toscc's, leaf of~$\Cc$ (\upo).
Each of them is the quotient of a Pascal automaton
and as such are complete (\uptr).
Ovals are the \ttscc; each of them has for unique successeur a square 
(\upt)
and `behaves in the same manner'
as a circuit of~$0$'s from this square (dotted lines):
that is, every vertex of a \ttscc is associated with a vertex of a 
$0$-circuit of the \toscc and
two associated vertices have the same behaviour:
their respective successor by~$0$ are associated and they have the 
same successor by a non-$0$ digit (\upf).

\begin{example}
  \figur{upc-exa-1} shows a simple but complete example of a \upau.
  The framed subautomata are the minimisation of Pascal automata~$\pascal{\{1,2\}}{3}{2}$
  on the top and~$\pascal{\{1,2,3,4\}}{5}{2}$ on the bottom.
    The two others non-trivial SCC's,~$\{B_2,C_2\}$ and~$\{D_2\}$, are reduced to~0-circuits.
  Each of them has successors in only one Pascal automaton.

  The dotted lines highlight (UP-4).
    The circuit~$(B_2,~C_2)$ is embeddable in the Pascal automaton~$\set{A,~B,~C}$ 
    with the map~$B_2\mapsto B$ and~$C_2 \mapsto C$.
  A similar observation can be made for the circuit~$(D_2)$.
  
\end{example}

\SmallPicture
\begin{figure}[ht!]
  \centering
  \VCCall{up_example}
  \caption{A complete example of the \upcr}
  \lfigure{upc-exa-1}
\end{figure}

Completeness and correctness of the \upcr are established as
  follows.
\begin{enumerate}
 \item Every \upsn is accepted by a \upau;
 \item The \upcr is stable by quotient;
 \item Every \upau accepts a \upsn.
\end{enumerate}
The first two steps ensure completeness for minimal automata (as
  every \brec set of numbers is accepted by a \emph{unique minimal
  automaton)}, the third one plays for correctness.

\subsection{Every \upsn is accepted by a \upau}
\label{s.upc-cpl}

\begin{proposition}
	\label{p.up-com-1}
  For every integers~$m$ and~$p$ and for every set~$R$ of residues
    there exists a \upau accepting~$\EpRm$.
\end{proposition}

\subsubsection{When the period divides a power of the base}
\lsection{upc-div-base}
Let~$\EpR$ be a periodic set of numbers such that~$p\div b^j$ for some~$j$.
An automaton accepting~$\EpR$ is obtained by a generalisation of the method
  for recognising if an integer written in base~$10$ is a multiple of~$5$,
  namely checking if its unit digit is~$0$ or~$5$:
  from~\equnm{val} follows:

\begin{lemma}\llemma{pdivbj}
  Let~$d$ be an integer such that~$d\div b^j$
    (and~$d\notdiv b^{j-1}$) and~$u$ in~$\Abe$ of length~$j$.
  Then,~$w$ in~$\Abe$ is such that~$\mdd{\val{w}}{\val{u}}$ \ifof
    $w=u\xmd v$ for a certain~$v$.
\end{lemma}

\rfigure{aut_1mod4_2} 
shows an example of such a construction.
\begin{figure}[ht!]
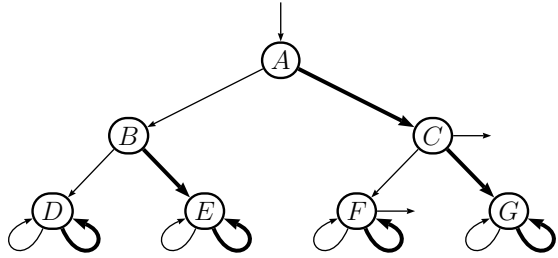

  \centering
  \VCCall{aut_1mod4_2}
  \caption{Automaton accepting integers congruent to 1 modulo 4 en base~2}
  \label{f.aut_1mod4_2}
\end{figure}

\subsubsection{The case of periodic sets of numbers}
\label{s.upc-per}

Let~$\EpR$ be a periodic set of numbers.
In contrast with \secti{def-pas},
$p$ and~$b$ are not supposed
to be coprime anymore.
Given a integer~$p$, there exist~$k$ and~$d$ such that
$p=k\xmd d$,
$k$ and~$b$ are coprime,
and~$d\div b^{j}$ for a certain~$j$.
The Chinese remainder theorem, a simplified version of which is given
below, allows to break the condition:
`being congruent to~$r$ modulo~$p$' into two simpler conditions.

\begin{theorem}[Chinese remainder theorem]\ltheorem{chinese}
Let~$k$ and~$d$ be two coprime integers.
Let~$r_k$,~$r_d$ be two integers.
There exists a unique integer~$r<k\xmd d$ such
that~$r \equiv r_k\xmd[k]$ and~$r \equiv r_d \xmd [d]$.

Moreover, for every~$n$ such that~$n \equiv r_k \xmd[k]$
and~$n \equiv r_d \xmd[d]$, we have~$n \equiv r \xmd[kd]$.
\end{theorem}

Let us assume for now that~$R$ is a singleton~$\{r\}$,
with~$r$ in~$\pterval$ and define~$r_d = (r \mod d)$
and~$r_k = (r \mod k)$.
\rtheorem{chinese} implies:
\begin{equation}
	\fa n\in\N \quantsp
	\mdp{n}{r}
	\ee\Longleftrightarrow\ee
	\mdk{n}{r_{k}}
	\e\text{and}\e
	\mdd{n}{r_{d}}
	\eqpnt
	\label{q.chi-1}
\end{equation}

The Pascal automaton~$\Pasc{r_{k}}{k}$ accepts the integers satisfying
$\mdk{n}{r_{k}}$ and an automaton accepting the integers satisfying
$\mdd{n}{r_{d}}$ can be defined from Lemma~\rlemma*{pdivbj}.
The \emph{product} of the two automata accepts the integers satisfying 
both equations of the right-hand side
of~\equnm{chi-1} and this is a \upau.

\begin{example} \label{e.5-12-2}
The following figures show the construction of an automaton accepting 
the set of representations in base~$2$ of the integers congruent 
to~$5$ modulo~$12$.
\figur{aut_2mod3_2m} shows~$\Pasc{2}{3}$, minimised for clarity,
\figur{aut_1mod4_2} shows an automaton accepting integers congruent 
to~$1$ modulo~$4$,
and
\figur{aut_5mod12_2} shows the product of the preceeding two, which 
accepts the required set of numbers.

\end{example}

\begin{figure}[ht!]
  \centering
  \VCCall{aut_2mod3_2m}
  \caption{The minimisation of~$\Pasc{2}{3}$}
  \label{f.aut_2mod3_2m}
\end{figure}

\begin{figure}[ht!]
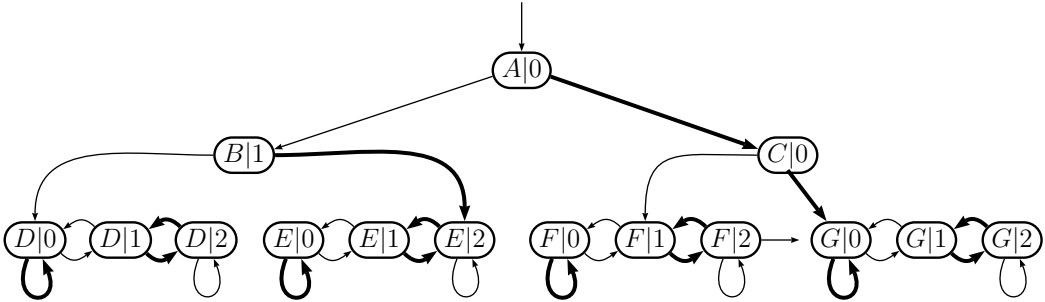

  \centering
  \VCCall{aut_5mod12_2}
  \caption{Automaton accepting integers congruent to 5 modulo 12 in base~2}
  \label{f.aut_5mod12_2}
\end{figure}

Let now~$R=\{r_1,r_2,\ldots,r_{\ell}\}$ be a subset of~$\pterval$.
In order to build an automaton~$\PSA{R}{p}$ that accepts~$\EpR$, let
$S=\{(r_{1,d},r_{1,k}),(r_{2,d},r_{2,k}),\ldots,(r_{\ell,d},r_{\ell,k})\}$ 
be the set of pairs~$(r_{i,d},r_{i,k})$ such that an integer~$n$ is 
congruent to~$r_{i}$ modulo~$p$ \ifof both
$\mdk{n}{r_{i,k}}$ and~$\mdd{n}{r_{i,d}}$.  

For every~$x<d$, let~$T_{x}=\Defi{r_{i,k}}{x=r_{i,d}}$, which means 
that if~$\mdd{n}{x}$ then~$n$ is in~$\EpR$ \ifof it is congruent to 
some~$t$ in~$T_{x}$ modulo~$k$.
It may be the case that for some~$x$,~$(k,T_{x})$ are not the 
canonical parameters for~$\PSN{T_{x}}{k}$.
An automaton that accepts~$\PSN{T_{x}}{k}$ is thus written as a 
generalised Pascal automaton~$\GPasc{T_{x}}{k}$.

The automaton~$\PSA{R}{p}$ consists then in a complete~$b$-tree of 
depth~$j$, whose~$b^{j}$ leaves are replaced by generalised Pascal 
automata.
More precisely, the word~$u$ of length~$j$ reaches the state~$q$ 
of~$\GPasc{T_{x}}{k}$ where~$\mdd{\val{u}}{x}$  
and~$q$ is defined by~$(0,0)\pathaut{u}{\GPasc{T_{x}}{k}}q$.
It is a routine to verify that the automaton constructed in such a 
way accepts~$\EpR$ and satisfies the \upcr.

\begin{example}
  \label{e.R-18-3}
\figur{R-18-3} shows the construction of an automaton accepting 
the set of representations in base~$3$ of the integers congruent 
to~$0$,~$2$,~$4$,~$5$, and~$9$ modulo~$18$.
The various parameters are:

~$p=18$,~$k=2$,~$d=9$,~$j=2$;
\ee
 $R=\{0,2,4,5,9\}$;
\ee
 $S=\{(0,0),(2,0),(4,0),(5,1),(0,1)\}$;

 $T_0=\{0,1\}$ since both~$(0,0)$ and~$(0,1)$ are in~$S$;

 $T_1=T_3=T_6=T_7=T_8=\emptyset$;
\ee
 $T_2=T_4=\{0\}$;
\ee
 and~$T_5=\{1\}$.
\end{example}

\begin{figure}[ht!]
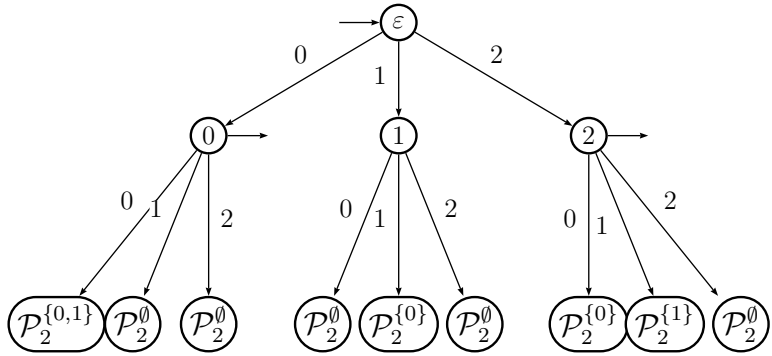

  \centering
  \VCCall{aut_02459mod18_3}
  \caption{Automaton accepting~$n\equiv 0,2,4,5,9 \mod 18$ in base~$3$}
  \label{f.R-18-3}
\end{figure}

\subsubsection{The case of arbitrary UP-sets of numbers}
\label{s.upc-u-per}

Let us denote by~$\Dc_m$ the automaton accepting words whose
  value is greater than~$m$.
It consists in a complete~$b$-tree~$T_m$ of depth~$\ceil{log_{b}(m)}$ plus a final
  sink state.
Every state may be labelled by the value of the word reaching it and
  it is final if its label is greater than~$m$.
Additionally, every leaf of~$T_m$ loops onto itself by reading a~$0$ and reaches the
  sink state by reading any other digit.
\begin{example}
  \figur{aut_sup5_2} shows~$\Dc_{5}$ (in base~$2$).
\end{example}
\begin{figure}[ht!]
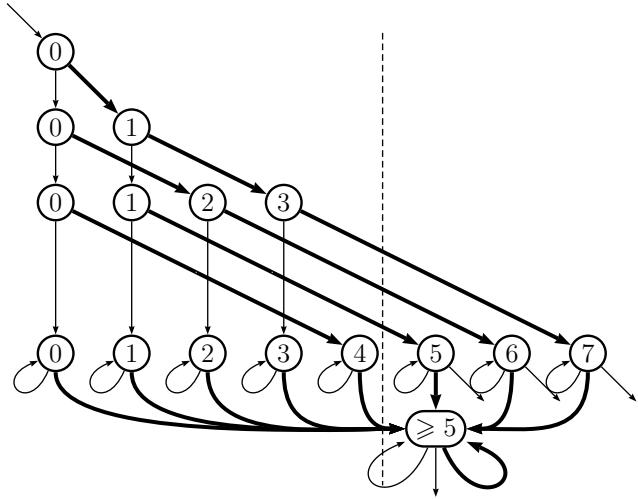

  \centering
  \VCCall{aut_sup5_2}
  \caption{Automaton accepting~$n \geq 5$ in base~$2$}
  \label{f.aut_sup5_2}
\end{figure}
Every~$\Dc_m$ is obviously a \upau.

An arbitrary \upsn~$\EpRm$ is accepted by
the product~$\PSA{R}{p} \times \GA{m}$, denoted by~$\UPSA{R}{p}{m}$.
The very special form of~$\Dc_{m}$ makes it immediate that this
product is a \upau, and this complete the proof of \propo{up-com-1}.

\begin{example}
The following figures show the construction of the 
  automaton~$\UPSA{0}{24}{1}$ accepting 
non-negative integers congruent to 0 modulo 24.

\rfigure{general_case_before_product1} shows the 
  automaton~$\Dc_1$,
\rfigure{general_case_before_product2} shows the 
  automaton~$\PSA{0}{24}$
and \rfigure{general_case_after_product} shows their 
  product,~$\UPSA{0}{24}{1}$.

\end{example}

\begin{table}[ht!]
  \noindent\hspace*{-\oddsidemargin}
  \begin{minipage}{\textwidth+2\oddsidemargin}
    \noindent\hfill\begin{minipage}{0.20\linewidth}~
    \end{minipage}%
    \hfill\begin{minipage}{0.60\linewidth}%
        \VCCall{general_case_before_product2}
        \captionof{figure}{$\PSA{0}{24}$}
        \label{f.general_case_before_product2}

    \end{minipage}\hfill~%
    
    \noindent\hfill\begin{minipage}{0.20\linewidth}
      \VCCall{general_case_before_product1}     
      \captionof{figure}{$\Dc_{1}$} 
      \label{f.general_case_before_product1}
    \end{minipage}%
    \hfill\begin{minipage}{0.60\linewidth}%
      \VCCall{general_case_after_product}
      \captionof{figure}{$\Dc_{1} \times \PSA{0}{24} = \UPSA{0}{24}{1}$}
      \label{f.general_case_after_product}                 
  \end{minipage}\hfill~%
  \end{minipage}
  
\end{table}

\subsection{The \upcr is stable by quotient}
\label{s.upc-quo}

\begin{proposition}\label{p.completeness2}
  If~$\Ac$ is a \upau, then every quotient of~$\Ac$ is also a \upau.
\end{proposition}

The \upcr relies on properties of SCC's that are stable by quotient.
The proof of \rproposition{completeness2} then consists essentially of proving that SCC's
are mapped into SSC's by the quotient.

\begin{lemma}\llemma{scc-quot} 
  Let~$\Ac$ and~$\Bc$ be two 
    deterministic\footnote{If the automaton is not deterministic, 
                            it would work as well, with the appropriate
                            definition of covering (see \cite{Saka09}).}
    finite automata, and~$\phi$ a covering~$\Ac \mapsto \Bc$.\nopagebreak
  
  Every SCC of~$\Bc$ is the quotient by~$\phi$ of a SCC of~$\Ac$.
  
\end{lemma}
\begin{proof}
  \begin{enumerate}
    \item \emph{ If two vertex x and y are strongly connected in~$\Ac$,
          then~$\phi(x)$ and~$\phi(y)$ are strongly connected in~$\Bc$}

    This is a direct consequence of the morphism.

    \item \emph{For every SCC~$C'$ of~$\Bc$, there exist a SCC~$C$ of~$\Ac$
      such that~${\phi(C)=C'}$.}

    Let ~$S$ its inverse image of~$C'$ by~$\phi$.

    \begin{enumerate}
%

      \item \emph{There exists a strongly connected set of states~$V$ contained 
      in~$S$, with no edge going from~$V$ to~$S\backslash V$.}

      One can order~$S$ with the reachability relation.
      Since~$S$ is finite, there is a minimal equivalency class.
      We denote this set of states by~$V$.

      \item \emph{$\phi(V)$ is equal to~$C'$.}

      For all vertex~$x$ in~$S$ and~$y'$ in~$C'$,~$x$ can reach (without 
      leaving~$S$) some vertex of~$\phi^{\uminus1}(y')$,
      a direct consequence of morphism.
      If~$x$ is taken in~$V$, it means that~$y'$ is in~$\phi(V)$, for all~$y'$ 
      in~$C'$, hence~$\phi(V)=C'$
    \end{enumerate}
    Let us denote by~$C$ the SCC containing~$V$.
    Since~$\phi(C)$ is strongly connected (from (1)) and contains~$C'$,
    then~$\phi(C)=C'$.

    \item \emph{For every SCC~$C'$ of~$\Bc$, there exist a SCC~$C$ of~$\Ac$
          such that~$\phi_{|C}$ induces a morphism from~$C$ to~$C'$.}

      We denote by~$\Gamma$ the set of SCC of~$\Ac$ whose image by~$\phi$ 
      is~$C'$.
      Since~$\Gamma$ is not empty(from (b)), there exists a 
      SCC~$C$ of~$\Gamma$ that cannot reach any other SCC of~$\Gamma$.
      It follows that every internal transition of~$C'$ is also internal 
      in~$C$, hence~$\phi_{|C}$ induces a morphism from~$C$ to~$C'$.
    \end{enumerate}
\vspace*{-1.25em}
\end{proof}

\begin{proof}[Proof of Proposition \ref{p.completeness2}] 
  Let~$\Ac$ be a \upau and~$\Bc$ one of its quotients.  
  The very definition of a quotient implies that~$\Bc$ satifies (UP-0).

  \rlemma{scc-quot} forces every Type 1 (resp. Type 2) SCC of~$\quot{\Ac}$ to be
    the quotient by~$\phi$ of a Type 1 (resp. Type 2) SCC of~$\Ac$. 
  Proving that~$\Bc$ satifies~(UP-1) up to~(UP-4),  is then immediate.
\end{proof}

\subsection{Every \upau accepts a \upsn}{\label{s.upc-cor}}

Let~$\Ac$ be a \upau and~$\Cond{\Ac}$ its condensation.
We call \emph{branch} of~$\Cond{\Ac}$ any path going from the root
  to a leaf using no loops. There is finitely many of them.
The inverse image by~$\gamma$ of a branch of~$\Cond{\Ac}$ define a subautomaton
of~$\Ac$.
Since a finite union of \upssn is still \up,
it is sufficient to prove the following statement.

\begin{proposition}\lproposition{correctness2}
  Let~$\Ac$ be a \upau and~$\Cond{\Ac}$ its condensation. 
  The inverse image by~$\gamma$ of a branch of~$\Cond{\Ac}$ accepts
    a \upsn.
\end{proposition}

\begin{proof}

Without loss of generality 
(that is, up to a finite number of elements in~$\ASNu{\Ac}$), 
one can assume that all final states belong to the SCC's of~$\Ac$.
Moreover, if~$\Ac$ has both a Type~1 and a Type~2 SCC, one can 
assume, by (UP-4) and up to the addition or subtraction of \emph{one} 
element in~$\ASNu{\Ac}$, that the \ttscc has the same final or 
non-final status as its image in the \toscc~$\Sc$ of~$\Ac$, and then, 
by minimisation, that~$\Ac$ has no \ttscc.

Let~$u$ be the shortest (and unique) word that sends the initial 
state~$i$, 
the root of~$\Cond{\Ac}$, into~$\Sc$: 
$i\pathaut{u}{\Ac}j$.
Let~$\Sc_{j}$ be the automaton obtained from~$\Sc$ by taking~$j$ as 
initial state.
By (UP-3) (and \propo{initial_shift})~$\Sc_{j}$ is a quotient of a Pascal 
automaton and accepts a periodic set of numbers~$\EpR$.

Every~$w$ in~$\compA$ is of the form~$w=u\xmd v$ with~$v$ 
in~$\CompAuto{\Sc_{j}}$.
Hence
$\val{w}= \val{u} + b^{\wlen{u}}\val{v}$
and~$w$ is in~$\compA$ \ifof
$\val{w}\jsgeq \val{u}$ and~$\val{w}$ belongs to~$b^{\wlen{u}}\EpR$.
Then,~$\ASNu{\Ac}$ is an \upsn of period~$p\xmd b^{\wlen{u}}$.
\end{proof}


\section{Conclusion and future work}

This work almost closes the complexity question raised by the
  Honkala's original paper~\cite{Honk86}.
The simplicity of the arguments in the proof should not hide that the
  difficulty was to make the proofs simple.
Two questions remain: getting rid, in~\theor{com-plx} of the minimality
  condition; or of the condition of determinism.

We are rather optimistic for a positive answer to the first one.
Since the minimisation of a DFA whose SCC's are simple cycles can be done
  in linear time~(\cf~\cite{AlmeZeit08}), it should be possible to verify 
  in linear time that the higher part of
  the \upcr (DAG and Type 2 SCC's) is satisfied by the minimised of 
  a given automaton without performing the whole minimisation.
It remains to find an algorithm deciding in linear time whether a
  given DFA has the same behaviour as a Pascal automaton.
This is the subject of still ongoing work of the authors.

On the other hand, defining a similar \upcr for nondeterministic
  automata seems to be much more difficult.
The criterion relies on the form and relations between SCC's,
  and the determinisation process is prone to destroy them.

\bibliography{%
  aux_bibliography,%
  \BIBINPUTSDIR Alexandrie-abbrevs,%
  \BIBINPUTSDIR Alexandrie-AC,%
  \BIBINPUTSDIR Alexandrie-DF,%
  \BIBINPUTSDIR Alexandrie-GL,%
  \BIBINPUTSDIR Alexandrie-MR,%
  \BIBINPUTSDIR Alexandrie-SZ%
}

\newcommand{\OneLetter}[1]{#1}
\begin{thebibliography}{10}

\bibitem{AhoHopcUllm74}
Alfred~V. Aho, John~E. Hopcroft, and Jeffrey~D. Ullman.
\newblock {\em The Design and Analysis of Computer Algorithms}.
\newblock Addison-Wesley, 1974.

\bibitem{AlloEtAL09}
Jean-Paul Allouche, Narad Rampersad, and Jeffrey Shallit.
\newblock Periodicity, repetitions, and orbits of an automatic sequence.
\newblock {\em Theoret. Comput. Sci.}, 410:2795--2803, 2009.

\bibitem{AlmeZeit08}
Jorge Almeida and Marc Zeitoun.
\newblock Description and analysis of a bottom-up dfa minimization algorithm.
\newblock {\em Inf. Process. Lett.}, 107(2):52--59, 2008.

\bibitem{BruyEtAl94}
V{\'e}ronique Bruy{\`e}re, Georges Hansel, Christian Michaux, and Roger
  Villemaire.
\newblock Logic and $p$-recognizable sets of integers.
\newblock {\em Bull. Belg. Soc. Math.}, 1:191--238, 1994.
\newblock Corrigendum, \textit{Bull. Belg. Soc. Math.}~1:577 (1994).

\bibitem{Cobh69}
Alan Cobham.
\newblock On the base-dependance of the sets of numbers recognizable by finite
  automata.
\newblock {\em Math. Systems Theory}, 3:186--192, 1969.

\bibitem{CormEtAl09}
Thomas~H. Cormen, Charles~E. Leiserson, Ronald~L. Rivest, and Clifford Stein.
\newblock {\em Introduction to {A}lgorithms (3. ed.)}.
\newblock MIT Press, 2009.

\bibitem{DuraFRigo11}
Fabien Durand and Michel Rigo.
\newblock On {C}obham's theorem, 2011.
\newblock {HAL}-00605375, to appear in \textit{{A}uto{M}ath{A} {H}andbook},
  (J-E. Pin, Ed.), {E.M.S.}.

\bibitem{Frou92}
Christiane Frougny.
\newblock Representation of numbers and finite automata.
\newblock {\em Math. Systems Theory}, 25:37--60, 1992.

\bibitem{FrouSaka10hb}
Christiane Frougny and Jacques Sakarovitch.
\newblock Number representation and finite automata.
\newblock in \textit{Combinatorics, Automata and Number Theory}, V.~Berth\'e,
  M.~Rigo (Eds), {Encyclopedia of Mathematics and its Applications 135},
  Cambridge Univ. Press (2010) {34--107}.

\bibitem{GinsSpan66}
Seymour Ginsburg and Edwin~H. Spanier.
\newblock Semigroups, {P}resburger formulas and languages.
\newblock {\em Pacif. J. Math.}, 16:285--296, 1966.

\bibitem{Honk86}
Juha Honkala.
\newblock A decision method for the recognizability of sets defined by number
  systems.
\newblock {\em RAIRO Theor. Informatics and Appl.}, 20:395--403, 1986.

\bibitem{LecoRigo10hb}
Pierre Lecomte and Michel Rigo.
\newblock Abstract numeration systems.
\newblock in \textit{Combinatorics, Automata and Number Theory}, V.~Berth\'e,
  M.~Rigo (Eds), {Encyclopedia of Mathematics and its Applications 135},
  Cambridge Univ. Press (2010) {108--162}.

\bibitem{LecoRigo01}
Pierre Lecomte and Michel Rigo.
\newblock Numeration systems on a regular language.
\newblock {\em Theory Comput. Syst.}, 34:27--44, 2001.

\bibitem{Lero05}
J{\'e}r{\^o}me Leroux.
\newblock A polynomial time {P}resburger criterion and synthesis for number
  decision diagrams.
\newblock In {\em Logic in Computer Science 2005 (LICS'2005)}, pages 147--156.
  IEEE Comp. Soc. Press, 2005.
\newblock New version at arXiv:cs/0612037v1.

\bibitem{Much03}
A.~Muchnik.
\newblock The definable criterion for definability in {P}resburger arithmetic
  and its applications.
\newblock {\em Theoret. Computer Sci.}, 290:1433--1444, 2003.
\newblock Late publication in a journal of a preprint (in russian) issued in
  1991.

\bibitem{Saka09}
Jacques Sakarovitch.
\newblock {\em Elements of Automata Theory}.
\newblock Cambridge University Press, 2009.
\newblock Corrected English translation of \emph{\'El\'ements de th\'eorie des
  automates}, Vuibert, 2003.

\end{thebibliography}

\end{document}